\documentclass[12pt]{article}
\usepackage[a4paper,margin=27mm]{geometry}
\usepackage{amsmath,amssymb,amsthm,mathtools,mathrsfs}
\usepackage[T1]{fontenc}
\usepackage{lmodern,microtype,enumitem}
\usepackage[backend=biber,style=alphabetic,sorting=nyt,giveninits=true,
  maxbibnames=99,maxalphanames=99,sortcites=false]{biblatex}
\DeclareNameAlias{author}{given-family}
\DeclareFieldFormat[article,incollection,misc]{title}{#1}
\DeclareFieldFormat[book]{title}{\mkbibemph{#1}}
\DeclareFieldFormat{journaltitle}{\mkbibemph{#1}\isdot}
\DeclareFieldFormat{booktitle}{\mkbibemph{#1}}
\DeclareFieldFormat{volume}{#1}
\DeclareFieldFormat{number}{#1}
\DeclareFieldFormat{series}{#1}
\DeclareFieldFormat[article]{series}{#1}
\DeclareFieldFormat{pages}{#1}
\DeclareFieldFormat{pagetotal}{#1\addspace pp\adddot}
\newbibmacro*{paper:author-title}{%
  \usebibmacro{begentry}%
  \printnames{author}\setunit{\addperiod\space}%
  \printfield{title}}
\newbibmacro*{paper:publication}{%
  \printlist{publisher}\setunit{\addcomma\space}%
  \printlist{location}\setunit{\addcomma\space}%
  \printfield{year}}
\DeclareBibliographyDriver{article}{%
  \usebibmacro{paper:author-title}\setunit{\addperiod\space}%
  \printfield{journaltitle}%
  \iffieldundef{series}{}{\setunit{\addspace}%
    \printtext[parens]{\printfield{series}}}%
  \setunit{\addcomma\space}\printfield{volume}%
  \iffieldundef{number}{}{\printtext[parens]{\printfield{number}}}%
  \iffieldundef{pagetotal}
    {\setunit{\addcolon}\printfield{pages}}
    {\setunit{\addcomma\space}\printfield{pagetotal}}%
  \setunit{\addcomma\space}\printfield{year}%
  \usebibmacro{finentry}}
\DeclareBibliographyDriver{book}{%
  \usebibmacro{paper:author-title}%
  \iffieldundef{volume}
    {\setunit{\addperiod\space}\printfield{series}}
    {\setunit{\addcomma\space}\printtext{volume\space}%
     \printfield{volume}%
     \iffieldundef{series}{}{\setunit{\addspace}%
       \printtext{of\space}\printfield{series}}}%
  \setunit{\addperiod\space}\usebibmacro{paper:publication}%
  \usebibmacro{finentry}}
\DeclareBibliographyDriver{incollection}{%
  \usebibmacro{paper:author-title}\setunit{\addperiod\space}%
  \printtext{In\space}\printfield{booktitle}%
  \iffieldundef{pages}{}{\setunit{\addcomma\space}%
    \printtext{pages\space}\printfield{pages}}%
  \setunit{\addperiod\space}\usebibmacro{paper:publication}%
  \usebibmacro{finentry}}
\DeclareBibliographyDriver{misc}{%
  \usebibmacro{paper:author-title}\setunit{\addperiod\space}%
  \printfield{howpublished}\setunit{\addcomma\space}%
  \printfield{year}\usebibmacro{finentry}}

\defbibenvironment{bibliography}
  {\list{\mkbibbrackets{\printfield{shorthand}}}
    {\settowidth{\labelwidth}{\mkbibbrackets{CSZ24}}%
     \setlength{\leftmargin}{\labelwidth}%
     \addtolength{\leftmargin}{\labelsep}%
     \setlength{\itemsep}{2pt}%
     \setlength{\parsep}{0pt}%
     }%
   \sloppy\clubpenalty4000\widowpenalty4000\sfcode`\.=1000\relax}
  {\endlist}
  {\item}

\usepackage[hidelinks]{hyperref}
\newtheorem{theorem}{Theorem}[section]
\newtheorem{lemma}[theorem]{Lemma}
\newtheorem{proposition}[theorem]{Proposition}
\newtheorem{corollary}[theorem]{Corollary}
\theoremstyle{remark}\newtheorem{remark}[theorem]{Remark}
\newcommand{\T}{\mathbb T}\newcommand{\R}{\mathbb R}\newcommand{\Z}{\mathbb Z}
\newcommand{\C}{\mathbb C}\newcommand{\N}{\mathbb N}
\newcommand{\mes}{\operatorname{mes}}\newcommand{\dist}{\operatorname{dist}}
\newcommand{\Spec}{\operatorname{Spec}}\newcommand{\DC}{\operatorname{DC}}
\newcommand{\Proj}{\operatorname{Proj}}
\newcommand{\diag}{\operatorname{diag}}

\newcommand{\HH}{\mathcal H}\newcommand{\GG}{\mathcal G}
\newcommand{\eps}{\varepsilon}

\numberwithin{equation}{section}
\title{Anderson localization for 1-d quasi-periodic Schr\"odinger operators with degenerate weights}
\author{Yingdu Dong\thanks{\raggedright School of Mathematical Sciences, Fudan University, Shanghai, China (202131130009@mail.bnu.edu.cn).}
\and Haoxuan Liu\thanks{\raggedright School of Mathematical Sciences, Fudan University, Shanghai, China (23110180021@m.fudan.edu.cn).}
\and Zuhong You\thanks{\raggedright School of Mathematical Sciences, Fudan University, Shanghai, China (23110180049@m.fudan.edu.cn).}
\and Xiaoping Yuan\thanks{\raggedright School of Mathematical Sciences, Fudan University, Shanghai, China (xpyuan@fudan.edu.cn).}}
\date{}
\begin{document}\maketitle
\begin{abstract}
We study Anderson localization for the quasi-periodic operators
\[
[Hu]_n=\frac{-u_{n+1}-u_{n-1}+\lambda v(x+n\rho)u_n}{w(y+n\alpha)},
\qquad n\in\Z,
\]
where $v,w$ are real analytic and $1$-periodic, $w\not\equiv0$ is nonnegative with real zeros, and $\lambda\in\R$. The frequencies $\rho$ and $\alpha$ sample the potential and the weight, respectively. For every prescribed Diophantine $\rho$, every fixed analytic potential and every fixed phase $(x,y)$ with $w(y)>0$, we prove pure point spectral type with exponentially localized eigenvectors in the region $|E|\geq E_*$, where $E_*$ is independent of the coupling, for every fixed $\lambda$ and almost every weight frequency $\alpha$ in the corresponding Diophantine slice. The result applies to the physical normal-mode equation of a Frenkel--Kontorova chain with quasi-periodic masses, with the substrate strength and analytic equilibrium fixed below the Aubry transition. High-energy localization is driven by the quasi-periodic mass disorder, and the abstract result also applies to constant or absent potentials. The weighted problem is unitarily equivalent to a singular Jacobi operator with unbounded off-diagonal coefficients. We establish its self-adjointness and a Schnol-type theorem. The proof combines a large-deviation estimate uniform in the phase of the prescribed frequency and in the coupling with a semi-algebraic orbit-avoidance estimate on each prescribed potential-frequency slice. We also establish full-spectrum localization at large coupling for every nonconstant analytic potential, at each fixed admissible phase and almost every frequency pair in the Diophantine class.
\end{abstract}
\section{Introduction}
Research on localization for discrete quasi-periodic Schr\"odinger operators has seen considerable advances over the past decades. Developments have been made in several directions: from one to higher dimensions, from single-frequency to multi-frequency, from perturbative to non-perturbative methods, and from Anderson localization to dynamical localization. Pioneering contributions were made by Sinai, Dinaburg, Chulaevsky, Fr\"ohlich, Spencer, Eliasson, Jitomirskaya, Bourgain, Goldstein, and Schlag, among others. Subsequently, this field has attracted a growing number of researchers, leading to numerous advances. We refer to \cite{Spe10} for historical remarks and to \cite{GYZ23,CSZ24,Ge24} for recent progress in several of these directions.

The present paper is concerned with a weighted version of this problem. More precisely, we study
\begin{equation}\label{weighted}
\begin{split}
[Hu]_n&=[H(\lambda,x,y,\rho,\alpha)u]_n\\
&=w^{-1}(y+n\alpha)\bigl(-u_{n+1}-u_{n-1}+\lambda v(x+n\rho)u_n\bigr).
\end{split}
\end{equation}
Here $(x,y)\in\T^2$ denotes the phase, where $\T=\R/\Z$. We write $\omega=(\rho,\alpha)\in\T^2$, where $\rho$ samples the potential $v$ and $\alpha$ samples the weight $w$. In the high-energy theorem, $\rho$ is prescribed and $\alpha$ varies. As a convention throughout the paper, for fixed $c_0>0$ and $A\geq3$ we denote by $\DC_{c_0,A}$ the set of frequencies satisfying
\begin{equation}\label{dc}
\|l\cdot\omega\|_{\T}\geq c_0|l|^{-A},\qquad l\in\Z^2\setminus\{0\},
\end{equation}
where $|l|=\max(|l_1|,|l_2|)$ and $\|\cdot\|_{\T}$ denotes the distance to the nearest integer. We assume that $\omega$ belongs to this fixed Diophantine class. The potential $v$ and the nonnegative weight $w$ are assumed to be real-valued, $1$-periodic analytic functions. Moreover, $w\not\equiv0$, and $w$ has real zeros. The weight $w$ represents the masses in the motivating lattice model and is therefore nonnegative. For example, $w(y+n\alpha)=\sin^2(2\pi(y+n\alpha))$ satisfies these assumptions.

Our principal result establishes high-energy localization for every prescribed Diophantine potential frequency $\rho$ and every fixed coupling, with an energy threshold independent of the coupling. The potential may be any real analytic function, including a constant function or one with zeros.

This result is motivated by the Frenkel--Kontorova (FK) chain with quasi-periodic masses. In the analytic regime considered here, its linearized vibration equation leads to the operator \eqref{weighted}, with the weight representing the particle masses. Once the substrate strength and a static equilibrium configuration are prescribed, both the potential in this equation and its frequency $\rho$ are fixed. The mass frequency $\alpha$ can nevertheless vary without changing the equilibrium, since the masses enter only the kinetic energy. The relevant localization result must therefore apply with $\rho$ prescribed and exclude only a null set of mass frequencies in the corresponding Diophantine class. The model and its reduction to \eqref{weighted} are developed in Section~\ref{fkapplication}.

A localization theorem valid for almost every frequency pair $(\rho,\alpha)$ does not by itself give this conclusion: its exceptional set may contain an entire prescribed slice $\{\rho\}\times\T$. We establish the frequency exclusion on each such Diophantine slice.

Before stating our localization results, we specify the phase and frequency conditions and the natural weighted Hilbert space in which $H$ has a self-adjoint closure.

A phase $y\in\T$ is called \emph{admissible} if $w(y)>0$. These assumptions imply two elementary but important facts:
\begin{enumerate}[label=$(\mathrm{\roman*})$]
\item for every fixed admissible phase $y$, after excluding a null set $\mathcal F(y)$ of frequencies $\alpha$, one has $w(y+n\alpha)>0$ for all $n\in\Z$;
\item for every $y\in\T$ and every irrational $\alpha$, there exists a subsequence $\{n_k\}$ such that $w(y+n_k\alpha)$ tends to zero.
\end{enumerate}
Indeed, (ii) follows from the density of irrational rotation orbits, while (i) follows from the finiteness of the zero set of $w$ and the countability of the frequencies for which the orbit of a fixed admissible phase hits a zero. Lemma~\ref{orbit} gives a polynomial lower bound along these orbits for almost every frequency.

The weight degeneracy makes \eqref{weighted} an unbounded operator. Fix $(x_0,y_0)\in\T^2$ with $y_0$ admissible, and then choose $\alpha$ such that $w(y_0+n\alpha)>0$ for every $n$. Then $H$ is naturally symmetric and densely defined on finitely supported sequences in the weighted Hilbert space
\[
\ell^2_w(\Z)=\left\{u:\sum_{n\in\Z}|u_n|^2w_n<\infty\right\},
\qquad w_n=w(y_0+n\alpha),
\]
with inner product
\[
(h,k)_w=\sum_{n\in\Z}h_n\overline{k_n}w_n.
\]
After the unitary change of variables $u'_n=w_n^{1/2}u_n$, the generalized problem is transformed to the Jacobi operator $H_1$ whose action is
\[
[H_1u']_n=-\frac{u'_{n-1}}{\sqrt{w_nw_{n-1}}}
-\frac{u'_{n+1}}{\sqrt{w_nw_{n+1}}}
+\lambda\frac{v(x_0+n\rho)}{w_n}u'_n.
\]
Thus the transformed operator is a singular Jacobi operator with unbounded off-diagonal coefficients. Since the orbit of $y_0$ comes arbitrarily close to the zeros of $w$, the coefficients $1/\sqrt{w_nw_{n\pm1}}$ become unbounded along subsequences.

For each fixed admissible phase $y_0$, define
\[
\begin{split}
\DC_{c_0,A}(y_0)=\bigl\{\omega\in\DC_{c_0,A}:\ &\text{there exists }r\in\N\text{ such that}\\
&\min_{y_i\in Z_w}\|y_0+n\alpha-y_i\|_\T
\geq r^{-1}|n|^{-A}\quad(n\ne0)\bigr\},
\end{split}
\]
where $Z_w$ is the zero set of $w$. Lemma~\ref{orbit} shows that this class has full measure in $\DC_{c_0,A}$.

For a prescribed potential frequency $\rho\in\T$ satisfying
\begin{equation}\label{rhodc}
\|k\rho\|_\T\geq c_0|k|^{-A},\qquad k\in\Z\setminus\{0\},
\end{equation}
write
\[
\DC^\rho_{c_0,A}(y_0)
=\{\alpha\in\T:(\rho,\alpha)\in\DC_{c_0,A}(y_0)\}.
\]
The main theorem keeps $\rho$ fixed and excludes only a null set of weight frequencies $\alpha$ in this class.

\begin{theorem}[High-energy localization at a fixed potential frequency]\label{mainhigh}
Consider the operator \eqref{weighted}, where $v$ is any real-valued, $1$-periodic analytic function, and $w$ is real-valued, $1$-periodic and analytic, with $w\geq0$, $w\not\equiv0$, and at least one real zero. Fix $c_0>0$, $A\geq3$ and a potential frequency $\rho$ satisfying \eqref{rhodc}. There exists a threshold
\[
E_*=E_*(v,w,c_0,A)=E_0(v,w,c_0,A)+4,
\]
independent of $\lambda$, where $E_0$ is from Proposition~\ref{ldthigh}, such that the following holds for every fixed $\lambda\in\R$. For any fixed $(x_0,y_0)\in\T^2$ with $y_0$ admissible, there is a set
\[
\mathcal N_\rho\subset\DC^\rho_{c_0,A}(y_0),\qquad \mes_1\mathcal N_\rho=0,
\]
such that, for every $\alpha\in\DC^\rho_{c_0,A}(y_0)\setminus\mathcal N_\rho$, the self-adjoint closure of $H(\lambda,x_0,y_0,\rho,\alpha)$ in $\ell^2_w(\Z)$ has pure point spectral type on $\{|E|\geq E_*\}$, and every eigenvector in this region satisfies
\[
|u_n|\leq C_u e^{-c|n|},\qquad n\in\Z,
\]
for a finite constant $C_u>0$ and some $c>0$. The same conclusions hold for the unitarily equivalent Jacobi operator $H_1$ on $\ell^2(\Z)$. The exceptional set may depend on $\rho,\lambda,x_0,y_0,v,w,c_0,A$, but is independent of the energy and the eigenvector.
\end{theorem}

As an application, Theorem~\ref{mainhigh} gives exponentially localized high-frequency vibrational modes around each prescribed analytic FK equilibrium, with the substrate strength fixed. The conclusion holds at each fixed admissible phase for almost every mass frequency in the corresponding Diophantine class; see Corollary~\ref{fkcor} in Section~\ref{fkapplication}.

We also prove full-spectrum localization at large coupling for a general nonconstant analytic potential, allowing both frequencies to vary.

\begin{theorem}[Large-coupling localization]\label{mainlarge}
Let $v$ be a nonconstant real-valued, $1$-periodic analytic function, and let $w$ satisfy the weight assumptions of Theorem~\ref{mainhigh}. Fix $c_0>0$ and $A\geq3$. There exists $\lambda_0=\lambda_0(v,w,c_0,A)>0$ such that, for every fixed $|\lambda|>\lambda_0$ and every fixed $(x_0,y_0)\in\T^2$ with $w(y_0)>0$, there is a set
\[
\mathcal N\subset\DC_{c_0,A}(y_0),\qquad \mes_2\mathcal N=0,
\]
such that, for every $\omega=(\rho,\alpha)\in\DC_{c_0,A}(y_0)\setminus\mathcal N$, the self-adjoint operator $H(\lambda,x_0,y_0,\rho,\alpha)$ has pure point spectral type on the whole spectrum, and every eigenvector satisfies
\[
|u_n|\leq C_u e^{-c|n|},\qquad n\in\Z,
\]
for a finite constant $C_u>0$ and some $c>0$. The same holds for $H_1$. The exceptional set may depend on $\lambda,x_0,y_0,v,w,c_0,A$, but is independent of the energy and the eigenvector.
\end{theorem}

Theorem~\ref{mainlarge} allows zeros and sign changes of $v$, and covers the bounded-energy region as well as the high-energy tails. Its frequency exclusion is two-dimensional. In Theorem~\ref{mainhigh}, the stronger frequency quantifier is obtained on the high-energy region for every fixed coupling. This is the form needed for the prescribed FK equilibrium in Corollary~\ref{fkcor}.

\begin{remark}[Nonvanishing potentials]\label{largecor}
If $v_*:=\min_\T|v|>0$, full-spectrum localization at a prescribed potential frequency follows directly from Theorem~\ref{mainhigh} whenever
\[
|\lambda|>\lambda_*:=\frac{2+E_*\|w\|_\T}{v_*}.
\]
Indeed, $v$ has a constant sign. Writing $\tau$ for the sign of $\lambda v$, for every finitely supported $u$ we have
\[
\tau(Hu,u)_w\geq(|\lambda|v_*-2)\|u\|_{\ell^2}^2
\geq\frac{|\lambda|v_*-2}{\|w\|_\T}\|u\|_w^2.
\]
The inequality extends from the core to the self-adjoint closure. Hence
\[
\Spec(\tau H)\subset
\left[\frac{|\lambda|v_*-2}{\|w\|_\T},\infty\right),
\]
so the entire spectrum lies in $\{|E|>E_*\}$. The conclusion holds for every fixed $|\lambda|>\lambda_*$, every prescribed $\rho$ and every fixed admissible phase, with the exceptional set of weight frequencies from Theorem~\ref{mainhigh}.
\end{remark}

We highlight three aspects of the proofs that address the unbounded coefficients and the prescribed potential frequency.

\paragraph{A Schnol theorem for unbounded Jacobi operators.}
The off-diagonal coefficients $(w_nw_{n+1})^{-1/2}$ of $H_1$ are unbounded, whereas in the singular quasi-periodic models of \cite{JY17,JL17,JSY21,Zha24} the unboundedness is confined to the diagonal. We establish essential self-adjointness on finitely supported sequences and prove a version of the Schnol theorem for this operator. Theorem~\ref{schnol} shows that its canonical spectral measure is carried by energies admitting polynomially bounded generalized eigenfunctions, and that the spectrum is the closure of these generalized eigenvalues. The polynomial lower bound on the weights in Lemma~\ref{orbit} allows these generalized eigenfunctions to be transferred between $H_1$ and the weighted equation. These results provide the spectral basis for the localization proof: exponential decay of the generalized eigenfunctions gives pure point spectral type in the stated energy region.

\paragraph{Uniform phase-section estimates.}
Keeping $\rho$ fixed requires control of the $y$-sections uniformly in $x$, since the translated potential phase $x_0+n\rho$ still varies with $n$. We adapt the Green-function multiscale method of \cite{Bou05,BGS02,JLS20} to obtain semi-algebraic exceptional sets satisfying
\[
\sup_{x\in\T}\mes_1\{y:(x,y)\in X_N(\lambda,E,\omega)\}<e^{-N^\kappa},
\qquad \lambda\in\R,\quad |E|\geq E_0,
\]
at every sufficiently large scale, with thresholds and degree bounds independent of $\lambda$. Outside these sets, the Green's functions have subexponential norm bounds and exponential off-diagonal decay.

The main issue is to preserve this uniformity when $|\lambda|/|E|$ is arbitrarily large. A diagonal rescaling of $E^{-1}(H_0-EW)$ keeps its analytic norm in $y$ bounded independently of the coupling, while its hopping remains of size at most $|E|^{-1}$. The nonconstancy of $w$ gives the initial section estimate, and a one-variable matrix Cartan estimate propagates it through the induction. For the semi-algebraic reduction, a polynomial approximation with error relative to $|v|$ preserves the same bounds after rescaling. Thus the initial step, the induction and the polynomial approximation all retain one uniform section estimate, giving an energy threshold independent of $\lambda$ for constant potentials as well as potentials with zeros.

\paragraph{Resonance exclusion at a prescribed frequency.}
After taking the union over finite candidate spectra and eliminating the energy, the preceding section bounds give a semi-algebraic resonance set $S\subset\T_\alpha\times\T_x\times\T_y$ with uniformly small $(\alpha,y)$-sections at each fixed $x$. We must estimate the weight frequencies for which
\[
(\alpha,x_0+n\rho,y_0+n\alpha)\in S
\]
at some position $R\leq n<2R$, where $R$ is a large orbit scale. For each $n$, this is an intersection with the graph $y=y_0+n\alpha$ in the section at $x=x_0+n\rho$. Applying planar transversality separately at each position produces a term of order $R^{-1}$; summing over all such positions loses the decay needed for frequency exclusion.

The geometric obstruction is concentration along lines: a set of zero area may contain many orbit points or have a long intersection with a prescribed line. In the two-particle problem of Bourgain and Kachkovskiy \cite[Section~5]{BK19}, small phase-space measure alone is therefore insufficient to control the distribution of bad boxes along a two-dimensional lattice orbit. Here the lattice is one-dimensional, but fixing $\rho$ requires control of prescribed line intersections in the $(\alpha,y)$-sections. The shared geometric strategy is to separate directions supporting long segments from those with short intersections. In their frequency-exclusion argument, Bourgain and Kachkovskiy combine a three-dimensional Kakeya estimate with semi-algebraic geometry to control exceptional directions \cite[Sections~7--8]{BK19}.

Here Lemma~\ref{fixedrhoavoid} combines a planar Kakeya-type direction estimate with orbit counting for the prescribed rotation. The planar estimate controls the directions supporting long segments in each section. The section itself varies with $n$, so we record the presence of a long segment of slope $n$ by a semi-algebraic condition on $(n/R,x_0+n\rho)$. The Diophantine condition on $\rho$ then bounds the number of such positions by $O(B^C R^{1-\sigma})$, where $B$ is the degree of $S$ and $\sigma=1/(A+2)$. Only these positions require the $R^{-1}$ transversality bound; at the others, short intersection intervals give a direct measure bound. This yields a factor $R^{-\sigma}$ and makes the excluded measures summable while $\rho$ remains fixed.

Since the construction includes all candidate energies, the resulting null set of weight frequencies is independent of the energy and the generalized eigenfunction, as required in Theorem~\ref{mainhigh}.

For the large-coupling result, the corresponding Green-function estimates give small total phase measure at every real energy. The standard two-frequency transversality estimate then gives the exclusion needed for Theorem~\ref{mainlarge}, and the same spectral argument yields full-spectrum localization.

The remainder of the paper is organized as follows. Section~2 presents the necessary notations and definitions, establishes an orbit estimate, and proves self-adjointness and the Schnol-type theorem. Section~3 establishes the uniform phase-section large-deviation estimates at high energy and derives their large-coupling counterpart. Section~4 proves the semi-algebraic estimates at a fixed potential frequency and the two localization theorems. Section~\ref{fkapplication} derives the FK phonon equation, states and proves Corollary~\ref{fkcor}, and discusses the physical parameters and the relation to earlier localization results.

\section{Preliminaries and a Schnol-type theorem}
We fix the notation and prove the orbit estimate, self-adjointness and the Schnol-type theorem used in the localization proof.

\subsection{Notations and conventions}
We use the following notation and conventions.
\begin{itemize}
\item $X\lesssim Y$ denotes the statement that $X\leq CY$ for some constant $C>0$. More generally, given parameters such as $a,\gamma$, $X\lesssim_{a,\gamma}Y$ denotes the statement that $X\leq C(a,\gamma)Y$ for some positive constant depending on $a$ and $\gamma$. For parameters which are fixed throughout this paper, we omit the dependence for brevity. For nonnegative quantities, $X\ll Y$ means $X\leq cY$ for a sufficiently small positive constant $c$, depending only on the fixed parameters.
\item We write $\mes_k$ for $k$-dimensional Lebesgue measure on $\R^k$ and normalized Lebesgue measure on $\T^k$. On an affine $k$-plane, $\mes_k$ denotes the measure induced by the Euclidean metric; on a line, it is arclength. We omit the subscript when the dimension is clear.
\item For a parameter such as $b$, the notation $b_\pm$ means $b\pm\eps$ for a fixed sufficiently small $\eps>0$, as specified by the context. Whenever $b_+$ is used below with $b<1$, choose $\eps<1-b$.
\item $\Lambda$ is often used to denote a subset of $\Z$. For a fixed $k\in\Z$,
\[
\Lambda_M(k)=\{n\in\Z:|n-k|\leq M\},\qquad \Lambda_M=\Lambda_M(0).
\]
The notation $|\Lambda|$ denotes the size of $\Lambda$, namely its number of sites. An $L$-interval means an interval containing $L$ sites, whereas $\Lambda_M(k)$ has radius $M$ and contains $2M+1$ sites. Integer parts of scales are taken when necessary.
\item Denote
\begin{align*}
H_0(\lambda,x,\rho)&=\lambda v(x+n\rho)\delta_{n,n'}+\Delta,\\
V(x,\rho)&=v(x+n\rho)\delta_{n,n'},\\
W(y,\alpha)&=w(y+n\alpha)\delta_{n,n'},\\
\HH(\lambda,E,x,y,\rho,\alpha)&=H_0(\lambda,x,\rho)-EW(y,\alpha),\\
\HH_N(\lambda,E,x,y,\rho,\alpha)&=R_N\HH(\lambda,E,x,y,\rho,\alpha)R_N^*,\\
G_N(\lambda,E,x,y,\rho,\alpha)&=\HH_N(\lambda,E,x,y,\rho,\alpha)^{-1},
\end{align*}
where $(\Delta u)_n=-u_{n+1}-u_{n-1}$ and $R_N$ is the coordinate restriction to $[-N,N]\cap\Z$. The calligraphic letter $\HH$ distinguishes the energy-dependent operator from the weighted self-adjoint operator $H=W^{-1}H_0$ in \eqref{weighted}. For simplicity, we omit some or all of the parameters when there is no confusion. Similarly, for any finite $\Lambda\subset\Z$, write $\HH_\Lambda=R_\Lambda\HH R_\Lambda^*$ and $G_\Lambda=\HH_\Lambda^{-1}$. Here $R_\Lambda^*$ extends a vector by zero.
\item Choose $h>0$ such that $v$ and $w$ have bounded holomorphic extensions to $\{z:|\operatorname{Im}z|<h\}$. For a $1$-periodic function $f$ with such an extension, write
\[
\|f\|_h=\sup_{|\operatorname{Im}z|<h}|f(z)|,
\qquad \|f\|_{\T}=\sup_{x\in\T}|f(x)|.
\]
Their Fourier coefficients satisfy $|\widehat v(k)|+|\widehat w(k)|\leq C_he^{-2\pi h|k|}$. Dependence on these analytic bounds and on $h$ is included in dependence on $v,w$.
\item For a matrix $A$, $\|A\|$ denotes its operator norm and $\|A\|_{\rm HS}$ its Hilbert--Schmidt norm. In exceptional-set definitions, we set $\|A^{-1}\|=+\infty$ when $A$ is singular.
\item Write $c_{00}(\Z)$ for the sequences that vanish outside a finite subset of $\Z$.
\item Since $w$ is real analytic and not identically zero, its real zero set is finite. We denote the zeros by
\[
Z_w=\{y_1,\ldots,y_m\}\subset\T
\]
Let $C_w$ denote the maximal order of these zeros. Since $w$ is real analytic and nonnegative, there exists $c_w>0$ such that
\begin{equation}\label{zeroorder}
w(y)\geq c_w\dist(y,Z_w)^{C_w}
=c_w\min_{1\leq i\leq m}\|y-y_i\|_{\T}^{C_w},\qquad y\in\T.
\end{equation}

\item A nontrivial sequence $u:\Z\to\C$ is a generalized eigenvector of $H$ if $Hu=Eu$, or equivalently $H_0u=EWu$, coordinatewise for some $E$, and
\[
|u(n)|\leq C_u(1+|n|)^{C'_u},\qquad n\in\Z,
\]
for some $C_u,C'_u>0$. In this case, $E$ is referred to as a generalized eigenvalue of $H$. For a given $C'>0$, let $\GG_{C'}(H)$ denote the set of $E\in\R$ for which there is a nontrivial solution satisfying $|u(n)|\leq C(1+|n|)^{C'}$ for some $C>0$. Denote the collection of these real generalized eigenvalues by
\[
\GG(H)=\bigcup_{C'>0}\GG_{C'}(H).
\]
The same notation is used for $H_1$ below.
\end{itemize}

We next give the quantitative orbit estimate used throughout the paper.

\begin{lemma}\label{orbit}
Let $w$ satisfy the assumptions of Theorem~\ref{mainhigh}, and fix an admissible phase $y_0$. For $r\in\N$, put
\[
\mathcal D_r(y_0)=\left\{\alpha\in\T:
\min_{1\leq i\leq m}\|y_0+n\alpha-y_i\|_\T
\geq r^{-1}|n|^{-A}\text{ for every }n\in\Z\setminus\{0\}\right\}.
\]
Then $\mathcal D_r(y_0)\subset\mathcal D_{r+1}(y_0)$ and
\[
\mes\left(\T\setminus\bigcup_{r\geq1}\mathcal D_r(y_0)\right)=0.
\]
For every $\alpha\in\mathcal D_r(y_0)$,
\begin{equation}\label{orbitbound}
w(y_0+n\alpha)\geq c(y_0,r)(1+|n|)^{-AC_w},\qquad n\in\Z,
\end{equation}
where $c(y_0,r)=\min\{w(y_0),c_wr^{-C_w}\}>0$. In particular,
\[
\DC_{c_0,A}(y_0)=\bigcup_{r\geq1}
\bigl[\DC_{c_0,A}\cap(\T\times\mathcal D_r(y_0))\bigr]
\]
has full measure in $\DC_{c_0,A}$.
\end{lemma}
\begin{proof}
Fix $y_0$. For $n\ne0$, the map $\alpha\mapsto y_0+n\alpha$ preserves normalized Lebesgue measure on $\T$. Hence
\[
\mes\left\{\alpha:\min_i\|y_0+n\alpha-y_i\|_\T
<r^{-1}|n|^{-A}\right\}\leq2mr^{-1}|n|^{-A}.
\]
Summing over $n\ne0$, we obtain
\[
\mes(\T\setminus\mathcal D_r(y_0))
\leq\frac{4m}{r}\sum_{n=1}^{\infty}n^{-A}\longrightarrow0
\qquad\text{as }r\to\infty.
\]
The sets $\mathcal D_r(y_0)$ are increasing, so their union has full measure. For $\alpha\in\mathcal D_r(y_0)$ and $n\ne0$, \eqref{zeroorder} gives
\[
w(y_0+n\alpha)\geq c_wr^{-C_w}|n|^{-AC_w}.
\]
Together with $w(y_0)>0$, this proves \eqref{orbitbound}. The last assertion follows by taking the product with $\T$ and intersecting with $\DC_{c_0,A}$.
\end{proof}

\subsection{Transformation to Jacobi operators}

We establish the self-adjointness of the weighted operator by passing
to a unitarily equivalent Jacobi operator, and then compare their
generalized eigenvectors.

Fix $\lambda\in\R$, $(x,y)\in\T^2$ with $y$ admissible, and
$\omega=(\rho,\alpha)$ with $\alpha\in\mathcal D_r(y)$ for some
$r\in\N$. Put $w_n=w(y+n\alpha)$. By Lemma~\ref{orbit},
$w_n>0$ for every $n$.

Initially, let $H$ act on $c_{00}(\Z)$ in $\ell^2_w(\Z)$.
Define
\[
U:\ell^2_w(\Z)\longrightarrow\ell^2(\Z),
\qquad (Uu)_n=w_n^{1/2}u_n.
\]
Since
\[
\|Uu\|_{\ell^2}^2=\sum_{n\in\Z}w_n|u_n|^2=\|u\|_w^2
\]
and $U^{-1}u'=W^{-1/2}u'$, the map $U$ is unitary.
Moreover, $U(c_{00}(\Z))=c_{00}(\Z)$. Thus, on $c_{00}(\Z)$,
\[
H_1:=UHU^{-1}=W^{-1/2}H_0W^{-1/2}
\]
is the symmetric Jacobi operator
\begin{equation}\label{h1}
[H_1u']_n
=-\frac{u'_{n-1}}{w_n^{1/2}w_{n-1}^{1/2}}
-\frac{u'_{n+1}}{w_n^{1/2}w_{n+1}^{1/2}}
+\lambda\frac{v(x+n\rho)}{w_n}u'_n.
\end{equation}

We next prove that $H_1$ is essentially self-adjoint.
By the two-sided Carleman criterion
\cite[Section~2.6, equation~(2.165)]{Tes00}, it suffices to show that
the sums of the reciprocals of the off-diagonal magnitudes diverge
at both ends, namely,
\[
\sum_{n\geq0}\sqrt{w_nw_{n+1}}
=\sum_{n\leq-1}\sqrt{w_nw_{n+1}}
=+\infty.
\]
If $\alpha$ is rational, the summands are positive and periodic,
so both sums diverge. If $\alpha$ is irrational, the continuous
function
\[
g(t)=\sqrt{w(t)w(t+\alpha)}
\]
is nonnegative and has strictly positive integral, since its zero
set is finite. Equidistribution of irrational rotation orbits
in both directions gives
\begin{align*}
\lim_{M\to\infty}\frac1M
\sum_{n=0}^{M-1}\sqrt{w_nw_{n+1}}
&=\int_\T g(t)\,dt>0,\\
\lim_{M\to\infty}\frac1M
\sum_{n=-M}^{-1}\sqrt{w_nw_{n+1}}
&=\int_\T g(t)\,dt>0.
\end{align*}
Hence both sums diverge. The criterion therefore implies that
$H_1$ is essentially self-adjoint, and its closure has the maximal
domain
\begin{equation}\label{domain}
\mathcal D(\overline H_1)
=\{f\in\ell^2(\Z):H_1f\in\ell^2(\Z)\},
\end{equation}
where $H_1f$ on the right is evaluated coordinatewise by \eqref{h1}.

Since $U$ is unitary, maps the initial domains onto each other,
and satisfies $H_1=UHU^{-1}$ there, taking closures gives
\[
\overline H=U^{-1}\overline H_1U,
\qquad
\mathcal D(\overline H)
=U^{-1}\mathcal D(\overline H_1).
\]
Thus $H$ is also essentially self-adjoint.
We henceforth denote these self-adjoint closures by $H$ and $H_1$.
They have the same spectrum and spectral types.

The same change of variables also gives a correspondence between
their generalized eigenvectors. Indeed, if $u'$ is a generalized
eigenvector of $H_1$ at energy $E$, then $u=W^{-1/2}u'$ satisfies
$H_0u=EWu$ coordinatewise. By \eqref{orbitbound},
\[
|u_n|
\leq c(y,r)^{-1/2}(1+|n|)^{AC_w/2}|u'_n|.
\]
Thus $u$ is polynomially bounded and is a generalized eigenvector
of $H$. Conversely, if $u$ is a generalized eigenvector of $H$,
then $u'=W^{1/2}u$ satisfies $H_1u'=Eu'$ coordinatewise and is
polynomially bounded because $w$ is bounded.

Since the off-diagonal coefficients in \eqref{h1} are nonzero,
the three-term recurrence shows that every finitely supported
sequence can be written as
\[
u=p(H_1)\delta_0+q(H_1)\delta_1
\]
for suitable polynomials $p,q$. These expressions are well defined
because $H_1$ leaves $c_{00}(\Z)$ invariant.
Consequently, $\{\delta_0,\delta_1\}$ is a cyclic pair for $H_1$.
Define the canonical spectral measure by
\[
\mu(B)
=\langle\mathbf1_B(H_1)\delta_0,\delta_0\rangle
+\langle\mathbf1_B(H_1)\delta_1,\delta_1\rangle,
\qquad B\subset\R\text{ Borel}.
\]
Cyclicity implies that every spectral measure of $H_1$ is absolutely continuous with respect to $\mu$.

The following Schnol theorem shows that $\mu$-almost every energy admits a polynomially bounded generalized eigenfunction and that every such generalized eigenvalue belongs to the spectrum. We give the proof for the present unbounded off-diagonal coefficients; see also \cite{DF22,Han19} for related results in bounded or long-range settings.

\begin{theorem}\label{schnol}
Under the assumptions of this subsection, let $\mu$ be the canonical spectral measure of $H_1$. Then the following statements hold.
\begin{enumerate}[label=$(\mathrm{\alph*})$]
\item For every $C'>1/2$,
\[
\mu(\R\setminus\GG_{C'}(H_1))=0.
\]
\item Every generalized eigenvalue of $H_1$ belongs to its spectrum:
\[
\GG(H_1)\subset\sigma(H_1).
\]
\end{enumerate}
Consequently, for every $C'>1/2$,
\[
\sigma(H_1)=\overline{\GG_{C'}(H_1)}=\overline{\GG(H_1)}.
\]
\end{theorem}
\begin{proof}
\emph{(a)} Fix $C'>1/2$. Denote by $\mathcal B$ the collection of Borel subsets of $\R$. For $B\in\mathcal B$ and $n,m\in\Z$, denote
\[
\mu_{n,m}(B)=\langle\mathbf1_B(H_1)\delta_m,\delta_n\rangle
\]
and
\[
\widetilde\mu(B)=\sum_{n\in\Z}t_n\mu_{n,n}(B),\qquad
t_n=c(1+|n|)^{-2C'},
\]
where $c>0$ is chosen so that $\sum_nt_n=1$. Then $\widetilde\mu$ is a Borel probability measure with $\widetilde\mu(B)=0$ if and only if $\mu(B)=0$. In one direction this follows because all $t_n$ are positive; in the other it follows from cyclicity of $\delta_0,\delta_1$. Hence $\widetilde\mu$ and $\mu$ are mutually absolutely continuous.

Now fix $n,m\in\Z$. Since
\[
|\mu_{n,m}(B)|\leq\mu_{n,n}(B)^{1/2}\mu_{m,m}(B)^{1/2},
\]
$\mu_{n,m}$ is also absolutely continuous with respect to $\widetilde\mu$. By the Radon--Nikodym theorem, there exists a measurable function $F_{n,m}\in L^1(\R,\widetilde\mu)$ with
\begin{equation}\label{rn}
\mu_{n,m}(B)=\int\mathbf1_B(E)F_{n,m}(E)\,d\widetilde\mu(E),\qquad B\in\mathcal B.
\end{equation}
Since $\mu_{n,n}\geq0$, we have $F_{n,n}\geq0$ $\widetilde\mu$-a.e.. Moreover, since $\widetilde\mu\ll\mu$,
\begin{equation}\label{nonzerocol}
F_{0,0}(E)+F_{1,1}(E)>0\quad\text{for }\widetilde\mu\text{-a.e. }E.
\end{equation}
By monotone convergence,
\begin{align*}
\widetilde\mu(B)&=\sum_{n\in\Z}t_n\mu_{n,n}(B)\\
&=\sum_{n\in\Z}t_n\int\mathbf1_B(E)F_{n,n}(E)\,d\widetilde\mu(E)\\
&=\int\mathbf1_B(E)\left(\sum_{n\in\Z}t_nF_{n,n}(E)\right)d\widetilde\mu(E).
\end{align*}
Thus,
\[
\sum_{n\in\Z}t_nF_{n,n}(E)=1\quad\text{for }\widetilde\mu\text{-a.e. }E,
\]
and consequently $F_{n,n}(E)\leq t_n^{-1}$ $\widetilde\mu$-a.e..

Then, for every $B\in\mathcal B$,
\begin{align*}
\left|\int\mathbf1_B(E)F_{n,m}(E)\,d\widetilde\mu(E)\right|
&=|\mu_{n,m}(B)|\\
&\leq\mu_{n,n}(B)^{1/2}\mu_{m,m}(B)^{1/2}\\
&=\left(\int\mathbf1_B(E)F_{n,n}(E)\,d\widetilde\mu(E)\right)^{1/2}
\left(\int\mathbf1_B(E)F_{m,m}(E)\,d\widetilde\mu(E)\right)^{1/2}\\
&\leq t_n^{-1/2}t_m^{-1/2}\widetilde\mu(B).
\end{align*}
Consequently,
\begin{equation}\label{densitygrowth}
|F_{n,m}(E)|\leq t_n^{-1/2}t_m^{-1/2}
\quad\text{for }\widetilde\mu\text{-a.e. }E.
\end{equation}
This follows from the preceding bound for all Borel sets. By countability, \eqref{densitygrowth} holds simultaneously for all $n,m$ on a set of full $\widetilde\mu$-measure.

On the other hand, \eqref{rn} implies that, for every bounded measurable function $f$,
\[
\langle f(H_1)\delta_m,\delta_n\rangle
=\int f(E)F_{n,m}(E)\,d\widetilde\mu(E).
\]
(Although $H_1$ is unbounded, the bounded Borel functional calculus is available for its self-adjoint closure.) For compactly supported bounded $g$, the function $Eg(E)$ is bounded. Since $\delta_n\in\mathcal D(H_1)$, the spectral theorem gives
\begin{align*}
\int g(E)EF_{n,m}(E)\,d\widetilde\mu(E)
&=\langle H_1g(H_1)\delta_m,\delta_n\rangle\\
&=\langle g(H_1)\delta_m,H_1\delta_n\rangle\\
&=-\frac{\langle g(H_1)\delta_m,\delta_{n-1}\rangle}{\sqrt{w_nw_{n-1}}}
-\frac{\langle g(H_1)\delta_m,\delta_{n+1}\rangle}{\sqrt{w_nw_{n+1}}}\\
&\hspace{1.5em}+\frac{\lambda v_n}{w_n}\langle g(H_1)\delta_m,\delta_n\rangle,
\end{align*}
where $v_n=v(x+n\rho)$. The last expression equals
\[
\int g(E)\left[-\frac{F_{n-1,m}(E)}{\sqrt{w_nw_{n-1}}}
-\frac{F_{n+1,m}(E)}{\sqrt{w_nw_{n+1}}}
+\frac{\lambda v_n}{w_n}F_{n,m}(E)\right]d\widetilde\mu(E).
\]
It follows, by equality of the corresponding locally finite measures and then countability of $(n,m)$, that on a fixed set of full $\widetilde\mu$-measure
\begin{equation}\label{densityrecurrence}
EF_{n,m}(E)=-\frac{F_{n-1,m}(E)}{\sqrt{w_nw_{n-1}}}
-\frac{F_{n+1,m}(E)}{\sqrt{w_nw_{n+1}}}
+\frac{\lambda v_n}{w_n}F_{n,m}(E).
\end{equation}
For every $E$ in this set satisfying \eqref{nonzerocol}, at least one of $m=0,1$ has $F_{m,m}(E)>0$. Choose that column and set
\[
u_E(n)=F_{n,m}(E).
\]
Then $u_E\not\equiv0$, \eqref{densityrecurrence} gives $H_1u_E=Eu_E$, and \eqref{densitygrowth} yields
\[
|u_E(n)|\leq t_m^{-1/2}t_n^{-1/2}
\leq C(1+|n|)^{C'}.
\]
Thus $E\in\GG_{C'}(H_1)$ for $\widetilde\mu$-a.e., and hence $\mu$-a.e., $E$.

\emph{(b)} Let $E\in\GG(H_1)$, and choose a nonzero solution
\[
H_1u=Eu,\qquad |u_n|\leq C_u(1+|n|)^p
\]
for some $p>0$. If $u\in\ell^2(\Z)$, then $H_1u=Eu\in\ell^2(\Z)$ coordinatewise. The maximal-domain description \eqref{domain} gives $u\in\mathcal D(H_1)$, so $E$ is an eigenvalue. We may therefore assume $u\notin\ell^2(\Z)$.

Put
\[
a_n=(w_nw_{n+1})^{-1/2},\qquad
M(N)=\sum_{|n|\leq N}|u_n|^2.
\]
The polynomial bound on $u$ gives
\[
M(N)\leq C_1(1+N)^{2p+1}.
\]
The coefficients $a_n$ need not be bounded, so we choose truncation boundaries where both $a_N$ and $a_{-N-1}$ are bounded.

If $\alpha$ is rational, the positive sequence $w_n$ is periodic and hence has a positive lower bound. In this case all $a_n$ are bounded. If $\alpha$ is irrational, the finiteness of the zero set of $w$ allows us to choose a nonempty open interval $I\subset\T$ and $c>0$ such that
\[
\min\{w(y+t),w(y+t+\alpha),w(y-t),w(y-t-\alpha)\}\geq c
\quad(t\in I).
\]
Indeed, the union of the zero sets of these four functions is finite, and one may choose the closure of $I$ inside its complement. By the density of $\{k\alpha:k\geq0\}$ in $\T$ and compactness, there is an integer $L_0\geq0$ such that
\[
\T=\bigcup_{k=0}^{L_0}(I-k\alpha).
\]
Thus every integer interval $[m,m+L_0]$ contains an $N$ with $N\alpha\in I$. At such an $N$,
\[
w_N,w_{N+1},w_{-N},w_{-N-1}\geq c,
\]
so $a_N,a_{-N-1}\leq c^{-1}$. Starting with a sufficiently large such $N_1$, choose $N_{j+1}$ to be the first such integer at least $N_j+3$. This gives $3\leq N_{j+1}-N_j\leq L_0+3$.

Consequently, in either case there are constants $c>0$, $L\geq3$ and positive integers $N_j$ such that $M(N_1)>0$ and
\begin{equation}\label{schnolcuts}
3\leq N_{j+1}-N_j\leq L,\qquad
a_{N_j}\leq c^{-1},\quad a_{-N_j-1}\leq c^{-1}.
\end{equation}
For rational $\alpha$, one can simply take $N_j=N_1+3(j-1)$.

Let $u^{(N)}=\mathbf1_{[-N,N]}u$. This vector belongs to $c_{00}(\Z)\subset\mathcal D(H_1)$ and has squared norm $M(N)$. Since $H_1u=Eu$, its truncation error is supported on the four sites $N,N+1,-N,-N-1$. By \eqref{h1} and \eqref{schnolcuts},
\begin{equation}\label{schnolcutoff}
\begin{aligned}
\|(H_1-E)u^{(N_j)}\|^2
={}&a_{N_j}^2\bigl(|u_{N_j}|^2+|u_{N_j+1}|^2\bigr)\\
&+a_{-N_j-1}^2\bigl(|u_{-N_j}|^2+|u_{-N_j-1}|^2\bigr)\\
\leq{}&c^{-2}\bigl[M(N_j+1)-M(N_j-1)\bigr].
\end{aligned}
\end{equation}

We claim that
\begin{equation}\label{schnolratio}
\liminf_{j\to\infty}
\frac{M(N_j+1)-M(N_j-1)}{M(N_j)}=0.
\end{equation}
Otherwise, there would be $\varepsilon>0$ and $j_0$ such that
\[
M(N_j+1)-M(N_j-1)\geq\varepsilon M(N_j)
\qquad(j\geq j_0).
\]
Since $M$ is nondecreasing and $N_j-N_{j-1}\geq3$, this implies
\[
M(N_j+1)\geq(1+\varepsilon)M(N_j-1)
\geq(1+\varepsilon)M(N_{j-1}+1).
\]
Iterating gives
\[
M(N_j+1)\geq M(N_{j_0}+1)(1+\varepsilon)^{j-j_0}
\qquad(j\geq j_0).
\]
On the other hand, the upper gap bound in \eqref{schnolcuts} gives
$N_j\leq N_{j_0}+L(j-j_0)$, and hence
\[
M(N_j+1)\leq C_2(1+j)^{2p+1}.
\]
The exponential lower bound contradicts this polynomial upper bound. This proves \eqref{schnolratio}.

Choose a subsequence $j_k$ along which the ratio in \eqref{schnolratio} is at most $k^{-2}$, and define
\[
\psi_k=\frac{u^{(N_{j_k})}}{\sqrt{M(N_{j_k})}}.
\]
Then $\psi_k\in\mathcal D(H_1)$, $\|\psi_k\|=1$, and \eqref{schnolcutoff} gives
\[
\|(H_1-E)\psi_k\|^2
\leq c^{-2}\frac{M(N_{j_k}+1)-M(N_{j_k}-1)}{M(N_{j_k})}
\leq c^{-2}k^{-2}\longrightarrow0.
\]
Weyl's criterion yields $E\in\sigma(H_1)$, proving (b).

Finally, for every $C'>1/2$, part (a) shows that
\[
0\leq\mu\bigl(\R\setminus\overline{\GG_{C'}(H_1)}\bigr)
\leq\mu\bigl(\R\setminus\GG_{C'}(H_1)\bigr)=0.
\]
Thus $\overline{\GG_{C'}(H_1)}$ is a closed set of full $\mu$-measure. Since every spectral measure of $H_1$ is absolutely continuous with respect to $\mu$, the spectrum of $H_1$ is the intersection of all closed sets of full $\mu$-measure. It follows that
\[
\sigma(H_1)\subset\overline{\GG_{C'}(H_1)}.
\]
Part (b) and the closedness of the spectrum give
\[
\overline{\GG_{C'}(H_1)}\subset\overline{\GG(H_1)}\subset\sigma(H_1),
\]
which proves the asserted equalities.
\end{proof}

\section[A large deviation theorem for the Green's functions]{\raggedright A large deviation theorem for the Green's \mbox{functions}}
We prove large-deviation estimates for the Green's function
\[
G_N(\lambda,E)=\bigl(R_N(H_0-EW)R_N^*\bigr)^{-1}
\]
by induction on the scale. We begin with a perturbation lemma.

\begin{lemma}[Perturbation]\label{perturb}
Consider operators $T=D+S$ and $T'=D'+S$, where $D$ and $D'$ are diagonal. For any finite subset $\Lambda\subset\Z$ of size $N$, let $G_\Lambda=T_\Lambda^{-1}$. The notations $T'_\Lambda$ and $G'_\Lambda$ are given similarly. Assume the following properties.
\begin{enumerate}[label=$(\mathrm{\roman*})$]
\item $G_\Lambda$ admits the Green's function estimates
\[
\|G_\Lambda\|\leq B(N),\qquad
|G_\Lambda(m,n)|\leq ae^{-\gamma|m-n|}\quad\text{for }|m-n|>K(N),
\]
where $a,\gamma>0$, $K(N)\geq0$, and $B(N)\geq\max(1,a)$.
\item $|(D'-D)(n,n)|\leq\eps$ for all $n\in\Lambda$.
\end{enumerate}
If $\eps NB^2(N)e^{2\gamma K(N)}\leq1/2$, then $T'_\Lambda$ is invertible and
\begin{align}
\|G'_\Lambda\|&\leq2B(N),\label{perturbnorm}\\
|G'_\Lambda(m,n)|&\leq(a+1)e^{-\gamma|m-n|}\quad\text{for }|m-n|>K(N).\label{perturbdecay}
\end{align}
\end{lemma}
\begin{proof}
We have
\[
T'_\Lambda=(I+(T'_\Lambda-T_\Lambda)G_\Lambda)T_\Lambda.
\]
The assumption gives $\eps B(N)\leq1/2$, so the first factor is invertible and
\begin{align*}
G'_\Lambda
&=G_\Lambda(I+(T'_\Lambda-T_\Lambda)G_\Lambda)^{-1}\\
&=G_\Lambda+\sum_{i=1}^{\infty}(-1)^iG_\Lambda\bigl((T'_\Lambda-T_\Lambda)G_\Lambda\bigr)^i.
\end{align*}
In particular, $\|G'_\Lambda\|\leq2B(N)$. Moreover, the hypotheses imply
\[
|G_\Lambda(m,n)|\leq B(N)e^{\gamma K(N)-\gamma|m-n|}\qquad(m,n\in\Lambda).
\]
Since $T'_\Lambda-T_\Lambda$ is diagonal, for every $i\geq1$ we have
\begin{align*}
&\left|\left[G_\Lambda\bigl((T'_\Lambda-T_\Lambda)G_\Lambda\bigr)^i\right](m,n)\right|\\
&\quad\leq\sum_{n_1,\ldots,n_i\in\Lambda}
|G_\Lambda(m,n_1)|\,|(D'-D)(n_1,n_1)|\,|G_\Lambda(n_1,n_2)|\cdots\\
&\hspace{9em}\cdots |(D'-D)(n_i,n_i)|\,|G_\Lambda(n_i,n)|\\
&\quad\leq N^i\eps^iB^{i+1}(N)e^{(i+1)\gamma K(N)}e^{-\gamma|m-n|}\\
&\quad\leq(\eps NB^2(N)e^{2\gamma K(N)})^ie^{-\gamma|m-n|}
\leq2^{-i}e^{-\gamma|m-n|}.
\end{align*}
The penultimate inequality uses $B(N)\geq1$ and $K(N)\geq0$. Summing the series proves \eqref{perturbdecay}.
\end{proof}

Now we establish the base for the induction.

\subsection{Initial step}
The weight $w$ is nonconstant because it has a zero but is not identically zero. Write $\Lambda_N=[-N,N]\cap\Z$. For $E\ne0$, put
\[
a_n=(\lambda/E)v(x+n\rho),\qquad d_n=1+|a_n|,\qquad Q=\diag(d_n^{1/2}),
\]
and define
\begin{equation}\label{workingmatrix}
B=Q^{-1}(E^{-1}\HH)Q^{-1},\qquad
\mathcal G_J=B_J^{-1}=E Q_JG_JQ_J.
\end{equation}
We write $\mathcal G_N=\mathcal G_{\Lambda_N}$. The initial and inductive estimates will be stated for $\mathcal G_N$; the relation to $G_N$ is given in \eqref{balancedtransfer} below.

\begin{lemma}[Initial estimate]\label{initialhigh}
Let $v$ be real analytic, possibly constant, $0<\kappa<1/10$ and $2\kappa<b<1$. There are $N_0\geq1$, $R_0\geq e^{20}$ and $C_1$, depending only on $v,w,\kappa,b$, such that, for every
\[
\lambda\in\R,\qquad |E|\geq R_0,\qquad
N_0\leq N\leq(\log|E|)^{1/(2\kappa)},\qquad\omega\in\T^2,
\]
there is a semi-algebraic set $X_N(\lambda,E,\omega)\subset\T^2$ of degree at most $N^{C_1}$ satisfying
\begin{equation}\label{initialsection}
\sup_{x\in\T}\mes_1\{y\in\T:(x,y)\in X_N\}<e^{-N^\kappa}.
\end{equation}
For all $(x,y)\notin X_N$,
\begin{align}
\|\mathcal G_N\|&<e^{N^b},\label{initialnorm}\\
|\mathcal G_N(m,n)|&<10e^{-3|m-n|/2},\qquad |m-n|>N/10.\label{initialdecay}
\end{align}
For fixed $\lambda$, these sets are jointly semi-algebraic in $(E,\omega,x,y)$. All constants are independent of $\lambda,E,\omega$.
\end{lemma}
\begin{proof}
The entries of $B$ are
\begin{equation}\label{balancedentries}
B(n,n)=\frac{a_n-w(y+n\alpha)}{1+|a_n|},\qquad
B(n,n+1)=-\frac1{E\sqrt{d_nd_{n+1}}}.
\end{equation}
Since $d_n\geq1$, we have $\|Q_J^{-1}\|\leq1$ and
\begin{equation}\label{balancedtransfer}
\begin{split}
G_J&=E^{-1}Q_J^{-1}\mathcal G_JQ_J^{-1},\\
\|G_J\|&\leq |E|^{-1}\|\mathcal G_J\|,\qquad
|G_J(m,n)|=\frac{|\mathcal G_J(m,n)|}{|E|\sqrt{d_md_n}}
\leq |E|^{-1}|\mathcal G_J(m,n)|.
\end{split}
\end{equation}
Let $\eta_N=e^{-N^{4\kappa/3}}$ and
\[
\widetilde X_N=\{(x,y):|B(n,n)|<\eta_N\text{ for some }|n|\leq N\}.
\]
The uniform scalar-level \L{}ojasiewicz estimate for $w$ \cite[Lemma~7.3]{Bou05} gives
\begin{equation}\label{scalarlevel}
\sup_{t\in\R}\mes_1\{y\in\T:|w(y)-t|<\eta\}\leq\eta^{c_1(w)}
\end{equation}
for some $c_1(w)>0$ and all sufficiently small $\eta>0$.

\emph{Case $(1)$: $|\lambda|\leq|E|$.}
Here $|a_n|\leq\|v\|_\T$. If $|B(n,n)|<\eta_N$, then
$|a_n-w(y+n\alpha)|<(1+\|v\|_\T)\eta_N$.
Applying \eqref{scalarlevel} for each fixed $x$ gives
\begin{equation}\label{highrawmeasure}
\sup_{x\in\T}\mes_1\{y:(x,y)\in\widetilde X_N\}
\leq C(2N+1)e^{-c_1(w)N^{4\kappa/3}}<e^{-2N^\kappa}.
\end{equation}
For $(x,y)\notin\widetilde X_N$, write $B_N=D_N+S_N$, where $D_N$ is diagonal. By \eqref{balancedentries},
\begin{equation}\label{diaghigh}
|D_N(n,n)|\geq e^{-N^{4\kappa/3}}>|E|^{-1/2},\qquad
\|S_N\|\leq2|E|^{-1},
\end{equation}
since $N^{2\kappa}\leq\log|E|$. Note that we have
\begin{align}
\mathcal G_N
&=(D_N+S_N)^{-1}\notag\\
&=D_N^{-1}(I+S_ND_N^{-1})^{-1}\label{neumannhigh}\\
&=D_N^{-1}+\sum_{i=1}^{\infty}(-1)^iD_N^{-1}(S_ND_N^{-1})^i.\label{neumannhighseries}
\end{align}
The series converges because
\[
2|E|^{-1}e^{N^{4\kappa/3}}
\leq2e^{-N^{2\kappa}+N^{4\kappa/3}}<\frac12
\]
for sufficiently large $N_0$. Combining \eqref{diaghigh} and \eqref{neumannhigh}, we have
\begin{equation}\label{rawnormhigh}
\|\mathcal G_N\|\leq2e^{N^{4\kappa/3}}.
\end{equation}
Moreover, by \eqref{diaghigh}, for $i\geq1$,
\begin{align*}
&\left|\left[D_N^{-1}(S_ND_N^{-1})^i\right](m,n)\right|\\
&\quad\leq\sum_{\substack{n_0=m,\,n_{2i+1}=n\\ |n_j|\leq N,\ j=1,\ldots,2i}}
|D_N^{-1}(n_0,n_1)|\,|S_N(n_1,n_2)|\,
|D_N^{-1}(n_2,n_3)|\cdots\\
&\hspace{9em}\cdots |S_N(n_{2i-1},n_{2i})|\,
|D_N^{-1}(n_{2i},n_{2i+1})|\\
&\quad\leq e^{N^{4\kappa/3}}(2N+1)^i|E|^{-i/2}e^{2i}e^{-2|m-n|}\\
&\quad\leq2^{-i}e^{N^{4\kappa/3}}e^{-2|m-n|}.
\end{align*}
Here $D_N^{-1}$ is diagonal and $|S_N(j,k)|\leq |E|^{-1}e^2e^{-2|j-k|}$. The last inequality follows from $(2N+1)e^2|E|^{-1/2}<1/2$, uniformly in the stated scale range after increasing $N_0$. Therefore, from \eqref{neumannhighseries}, we obtain that, for $(x,y)\notin\widetilde X_N$,
\begin{equation}\label{rawdecayhigh}
|\mathcal G_N(m,n)|
\leq2e^{N^{4\kappa/3}}e^{-2|m-n|}<e^{-7|m-n|/4},\qquad |m-n|>N/10.
\end{equation}

\emph{Case $(2)$: $|\lambda|>|E|$.}
Put $K_w=\|w\|_\T+2$. If $|a_n|>K_w$, then
\[
|B(n,n)|\geq\frac{|a_n|-\|w\|_\T}{1+|a_n|}
\geq\frac2{1+K_w}.
\]
Otherwise $|B(n,n)|<\eta_N$ implies
$|a_n-w(y+n\alpha)|<(1+K_w)\eta_N$.
Thus \eqref{scalarlevel} again gives \eqref{highrawmeasure}, uniformly in $\lambda/E$.
The diagonal and hopping bounds in \eqref{diaghigh} are unchanged, so the Neumann expansion and matrix-product estimates of Case $(1)$ apply. This argument includes constant $v$.

We next make the exceptional set semi-algebraic. Set $\epsilon_N=e^{-N^3}$. Choose real polynomials $v_1,w_1$ of degree less than $N^4$ such that
\begin{equation}\label{polynomial}
\|v-v_1\|_\T<\epsilon_N,\qquad
\|w-w_1\|_\T<\epsilon_N,
\end{equation}
and, in addition,
\begin{equation}\label{relativepoly}
|v(t)-v_1(t)|\leq\epsilon_N|v(t)|,\qquad 0\leq t\leq1.
\end{equation}
For constant $v$, take $v_1=v$. Otherwise, let $z_1,\ldots,z_s$ be the distinct zeros of $v$ in the closed interval $[0,1]$, with multiplicities $m_1,\ldots,m_s$, and put
\[
p(t)=\prod_{j=1}^s(t-z_j)^{m_j},\qquad u(t)=\frac{v(t)}{p(t)},
\]
with $p=1$ if there are no zeros. The quotient $u$ extends analytically across the zeros of $v$ and is nonzero on $[0,1]$. For sufficiently large $N$, analytic polynomial approximation gives a real polynomial $u_1$ of degree less than $N^4-\deg p$ such that
\[
\|u-u_1\|_{[0,1]}<\frac{\epsilon_N}{2}
\min\left\{\min_{[0,1]}|u|,\frac1{1+\|p\|_{[0,1]}}\right\}.
\]
Then $v_1=pu_1$ satisfies both estimates for $v$ in \eqref{polynomial}--\eqref{relativepoly}. Moreover, $u_1$ has the same sign as $u$ on $[0,1]$, so $v_1$ has the same zeros and multiplicities as $v$ on this interval.

The polynomials are evaluated at fractional parts of shifted phases. Define
\[
\HH_N^{(1)}=\Delta_N+\lambda V_{1,N}-EW_{1,N},\qquad
a_{n,1}=(\lambda/E)v_1(x+n\rho),\qquad d_{n,1}=1+|a_{n,1}|,
\]
and $Q_1=\diag(d_{n,1}^{1/2})$. By \eqref{relativepoly},
\begin{equation}\label{relativeweights}
1-\epsilon_N\leq d_{n,1}/d_n\leq1+\epsilon_N.
\end{equation}
For the intermediate matrix $B_N'=Q_N^{-1}(E^{-1}\HH_N^{(1)})Q_N^{-1}$, the hopping is the same as for $B_N$, and
\begin{equation}\label{balancedpolyerror}
|(B_N'-B_N)(n,n)|
\leq\frac{\epsilon_N|a_n|+\epsilon_N}{1+|a_n|}
\leq\epsilon_N.
\end{equation}
Write $G_N^{(1)}=(\HH_N^{(1)})^{-1}$ when the inverse exists, and set
\begin{equation}\label{polyweightedgreen}
\mathcal G_N^{(1)}=E Q_{1,N}G_N^{(1)}Q_{1,N}
=\Theta_N(B_N')^{-1}\Theta_N,\qquad \Theta_N=Q_{1,N}Q_N^{-1}.
\end{equation}
Multiplication by $\Theta_N$ or $\Theta_N^{-1}$ changes norms and individual entry bounds by factors $1+O(\epsilon_N)$, uniformly in $\lambda,E$. Define $X_N$ by
\begin{equation}\label{polygoodinitial}
\begin{split}
X_N^c=\bigl\{(x,y):\ &\det\HH_N^{(1)}\ne0,\quad
\|\mathcal G_N^{(1)}\|_{\rm HS}<4(2N+1)e^{N^{4\kappa/3}},\\
&|\mathcal G_N^{(1)}(m,n)|<2e^{-3|m-n|/2}
\text{ for all }|m-n|>N/10\bigr\}.
\end{split}
\end{equation}
These conditions have a polynomial description of degree at most $N^{C_1}$. Indeed,
\begin{equation}\label{weightedpolysquares}
\begin{split}
\|\mathcal G_N^{(1)}\|_{\rm HS}^2
&=E^2\sum_{m,n}d_{m,1}d_{n,1}|G_N^{(1)}(m,n)|^2,\\
|\mathcal G_N^{(1)}(m,n)|^2
&=E^2d_{m,1}d_{n,1}|G_N^{(1)}(m,n)|^2.
\end{split}
\end{equation}
Partition according to the integer parts of the shifted phases and the signs of $E$ and $v_1(x+n\rho)$. For $v\equiv0$ the latter sign is constant; otherwise its changes occur only at the fixed, finitely many zeros of $v$. These partitions therefore have polynomially many pieces. On each piece, Cramer's rule and clearing squared determinant denominators and powers of the nonzero $E$ turn \eqref{polygoodinitial} into polynomial inequalities. This also gives the joint description in $(E,\omega,x,y)$ for fixed $\lambda$.

Apply Lemma~\ref{perturb} to $B_N$ and $B_N'$, and use \eqref{relativeweights}--\eqref{polyweightedgreen}. The strict margin in \eqref{rawdecayhigh} gives the entry threshold in \eqref{polygoodinitial}, while
$\|A\|_{\rm HS}\leq\sqrt{2N+1}\|A\|$ gives its norm threshold. Thus $X_N\subset\widetilde X_N$ and \eqref{highrawmeasure} proves \eqref{initialsection}.
Applying the same perturbation argument in the reverse direction proves \eqref{initialnorm}--\eqref{initialdecay}. Both perturbations are justified by
\[
\epsilon_N(2N+1)\bigl(5(2N+1)e^{N^{4\kappa/3}}\bigr)^2e^{2N/5}\longrightarrow0,
\]
and $4\kappa/3<b$ absorbs the polynomial factor in the norm.
\end{proof}

\begin{remark}
By \eqref{balancedtransfer}, the same norm and decay bounds hold for $G_N$ outside $X_N$. We retain the estimates for $\mathcal G_N$ to carry out the induction. The relative approximation \eqref{relativepoly} makes \eqref{balancedpolyerror} uniform in $\lambda$ at every scale.
\end{remark}

We next collect the technical results needed for the induction.

\subsection{Preparations}
We first prove the resolvent iteration used here and in Section~4, following \cite[(10.32) and Lemma~10.33]{Bou05} and \cite[Section~3]{JLS20}.

\begin{lemma}[Resolvent iteration]\label{gluing}
Let $T=D+S$ be a self-adjoint operator on $\ell^2(\Z)$, where $D$ is diagonal and $S$ satisfies
\[
|S(m,n)|\leq Ce^{-\gamma_0|m-n|}\qquad(m\ne n).
\]
Write $G_J=T_J^{-1}$ when the inverse exists. Fix $0<b<1$ and $0<\gamma<\gamma_0/2$. Let $\Lambda\subset\Z$ with $|\Lambda|=N$. Suppose that each $k\in\Lambda$ belongs to an $M$-interval $I_k\subset\Lambda$ satisfying $\dist(k,\Lambda\setminus I_k)\geq M/4$, and that
\[
\|G_{I_k}\|\leq Ce^{M^b},\qquad
|G_{I_k}(m,n)|\leq Ce^{-\gamma|m-n|}
\quad\text{for }|m-n|>M/10.
\]
Then, for sufficiently large $M$, $G_\Lambda$ exists and
\begin{align}
\|G_\Lambda\|&\leq CMe^{M^b},\label{gluenorm}\\
|G_\Lambda(m,n)|&\leq
e^{-\left[\gamma-C\left(\frac{\log M}{M}+\frac MN\right)\right]|m-n|},
\qquad |m-n|>N/10.\label{glueentry}
\end{align}
The constants and the lower bound on $M$ depend only on the fixed bounds above and on $b,\gamma_0,\gamma$, uniformly for $\gamma$ in compact subsets of $(0,\gamma_0/2)$.
\end{lemma}
\begin{proof}
\emph{$(1)$ Norm estimate.} The resolvent identity reads
\begin{equation}\label{glueresolvent}
G_\Lambda(m,n)=G_{I_m}(m,n)\mathbf1_{I_m}(n)
-\sum_{\substack{m_1\in I_m\\n_1\in\Lambda\setminus I_m}}
G_{I_m}(m,m_1)S(m_1,n_1)G_\Lambda(n_1,n).
\end{equation}
Here the local term is zero outside $I_m$; the global inverse is justified below. The two local estimates imply
\[
|G_{I_m}(m,m_1)|\leq Ce^{M^b+\gamma M/10}e^{-\gamma|m-m_1|}.
\]
Since $|m-n_1|\geq M/4$, the total coefficient of the global terms in \eqref{glueresolvent} satisfies
\[
\begin{split}
q&:=\sup_{m\in\Lambda}
\sum_{\substack{m_1\in I_m\\n_1\in\Lambda\setminus I_m}}
|G_{I_m}(m,m_1)S(m_1,n_1)|\\
&\leq CM e^{M^b+\gamma M/10}
\sum_{d\geq M/4}e^{-\gamma d}
\leq CM e^{M^b-3\gamma M/20}<\frac12.
\end{split}
\]
The sums over distances are over integers. The same local restriction formula applied to $T_\Lambda\psi=0$ gives $\|\psi\|_\infty\leq q\|\psi\|_\infty$, so $G_\Lambda$ exists. Summing absolute values in \eqref{glueresolvent} over $n$ and taking the maximum over $m$ gives
\[
(1-q)\sup_m\sum_n|G_\Lambda(m,n)|\leq CMe^{M^b}.
\]
By symmetry the column sums have the same bound, so Schur's test proves \eqref{gluenorm}.

\emph{$(2)$ Off-diagonal estimate.} If $|m-n|>M$, then $n\notin I_m$, and the resolvent identity becomes
\begin{equation}\label{glueoffidentity}
G_\Lambda(m,n)=-\sum_{\substack{m_1\in I_m\\n_1\in\Lambda\setminus I_m}}
G_{I_m}(m,m_1)S(m_1,n_1)G_\Lambda(n_1,n).
\end{equation}
Put $d=|m-n_1|\geq M/4$. For $|m-m_1|\leq M/10$, we have $|m_1-n_1|\geq3d/5$, so the contribution of these sites to the coefficient is at most $CMe^{M^b-3\gamma_0 d/5}\leq CMe^{-\gamma d}$. For $|m-m_1|>M/10$, use the local decay and
\[
e^{-\gamma|m-m_1|-\gamma_0|m_1-n_1|}
\leq e^{-\gamma d}e^{-(\gamma_0-\gamma)|m_1-n_1|}.
\]
Summing over $m_1$ in \eqref{glueoffidentity} therefore gives
\begin{equation}\label{glueonestep}
|G_\Lambda(m,n)|\leq CM
\sum_{\substack{n_1\in\Lambda\\|m-n_1|\geq M/4}}
e^{-\gamma|m-n_1|}|G_\Lambda(n_1,n)|.
\end{equation}
Fix $m,n$ with $|m-n|>N/10$, put $n_0=m$, and take $r=\lceil4|m-n|/M\rceil$. Iterate \eqref{glueonestep} while $|n_i-n|>M$, at most $r$ times. The terms ending after $j$ steps, $0\leq j\leq r$, are bounded by the nested sum
\[
(CM)^j\sum_{n_1}^{*}\cdots\sum_{n_j}^{*}
e^{-\gamma(|m-n_1|+\cdots+|n_{j-1}-n_j|)}|G_\Lambda(n_j,n)|.
\]
The stars require $n_i\in\Lambda$, $|n_i-n_{i-1}|\geq M/4$, and $|n_i-n|>M$ for $i<j$; the last point satisfies $|n_j-n|\leq M$ unless $j=r$. The case $j=0$ is the initial term when $|m-n|\leq M$.

Set $s=12\log M/M$ and $\gamma_1=\gamma-s>0$. With $\ell_i=|n_i-n_{i-1}|$, summing successively over the indices gives
\[
(CM)^j\sum_{n_1}^{*}\cdots\sum_{n_j}^{*}e^{-s\sum_{i=1}^j\ell_i}
\leq\left(2CM\sum_{\ell\geq M/4}e^{-s\ell}\right)^j
\leq\left(\frac{C}{M\log M}\right)^j\leq2^{-j}.
\]
We now distinguish two cases.

\emph{Case $1$: $j<r$.} Since $|n_j-n|\leq M$, we have $\sum_{i=1}^j\ell_i\geq|m-n|-|n_j-n|\geq|m-n|-M$. By \eqref{gluenorm}, these terms contribute at most
\[
CMe^{M^b}e^{-\gamma_1(|m-n|-M)}\sum_{j=0}^{r-1}2^{-j}.
\]
\emph{Case $2$: $j=r$.} Whether or not the last point is close to $n$, the accumulated distance satisfies
\[
\sum_{i=1}^r\ell_i\geq\frac{rM}{4}\geq|m-n|.
\]
The same estimate bounds these terms by $CMe^{M^b}e^{-\gamma_1|m-n|}2^{-r}$. Combining the two cases yields
\[
|G_\Lambda(m,n)|\leq CMe^{M^b+\gamma M}
e^{-(\gamma-12\log M/M)|m-n|}.
\]
For $|m-n|>N/10$, the prefactor is absorbed into the $CM/N$ loss, proving \eqref{glueentry}.
\end{proof}

\begin{proposition}[Cartan estimate, {\cite[Proposition~14.1]{Bou05}}]\label{cartan}
Let $A(s)$ be a self-adjoint $N\times N$ matrix function of a real parameter $s\in[-a,a]$, where $a>0$. Assume the following conditions.
\begin{enumerate}[label=$(\mathrm{\roman*})$]
\item $A(s)$ is real analytic and has a holomorphic extension to
\[
\mathcal D_{a,a}=\{z\in\C:|\Re z|<a,\ |\Im z|<a\},
\]
with
\[
\sup_{z\in\mathcal D_{a,a}}\|A(z)\|<B_1,\qquad B_1\geq1.
\]
\item For each $s\in[-a,a]$, there is a subset $\Lambda\subset[1,N]$ such that
\[
|\Lambda|<M,\qquad
\|(R_{[1,N]\setminus\Lambda}A(s)R_{[1,N]\setminus\Lambda}^{*})^{-1}\|<B_2,
\qquad B_2\geq1.
\]
\item For some $B_3\geq1$,
\[
\mes_1\{s\in[-a,a]:\|A(s)^{-1}\|>B_3\}
<10^{-3}a(1+B_1)^{-1}(1+B_2)^{-1}.
\]
\end{enumerate}
If $M\geq1$ and $0<\eps<(1+B_1+B_2)^{-10M}$, then
\[
\begin{aligned}
&\mes_1\{s\in[-a/2,a/2]:\|A(s)^{-1}\|>\eps^{-1}\}\\
&\quad\leq C\exp\left[-c_2\frac{\log\eps^{-1}}
{M\log(M+B_1+B_2+B_3)}\right],
\end{aligned}
\]
where $C,c_2>0$ depend only on the fixed $a$ and $B_1$.
\end{proposition}

\begin{lemma}[Coupling lemma, {\cite[Lemma~2.1]{Bou02}}]\label{coupling}
Suppose $0<\tau<1$ and $0<b<1$. Let $A$ be a self-adjoint $N\times N$ matrix satisfying
\[
|A(m,n)|\leq Ce^{-\gamma_0|m-n|},\qquad 1\leq m,n\leq N,
\]
where $C,\gamma_0>0$ are fixed constants and $N$ is a large integer. Assume that for any interval $\Lambda\subset[1,N]$ with $|\Lambda|=L>N^\tau$, there exists a bound on the inverse
\[
\|A_\Lambda^{-1}\|\leq e^{L^b}.
\]
Fixing an integer $M$ such that $N^\tau<M<2N^\tau$, a size-$M$ interval $\Lambda\subset[1,N]$ is called a ``good'' one if
\[
|A_\Lambda^{-1}(m,n)|\leq e^{-\gamma|m-n|}
\quad\text{for all }m,n\in\Lambda\text{ with }|m-n|>M/10,
\]
where $0<\gamma<\gamma_0/10$. Assume that there are at most $N^b/M$ disjoint bad $M$-intervals in $[1,N]$.
Under the above assumptions, it holds that
\[
|A^{-1}(m,n)|\leq e^{-(\gamma-N^{-c_3})|m-n|},\qquad |m-n|>N/10,
\]
where $c_3=c_3(\tau,b)>0$.
\end{lemma}

\begin{lemma}[Semi-algebraic orbit counting, {\cite[Corollary~9.7]{Bou05}}]\label{counting}
Suppose $S\subset\T^\nu$ is a semi-algebraic set of degree $B$ with $\mes_\nu S<\eta$. Fix $c_0>0$ and $A\geq\nu+1$. Let $N$ be a large number satisfying
\[
\log B\ll\log N<\log\eta^{-1}.
\]
Assume that $\omega\in\T^\nu$ satisfies
\[
\|l\cdot\omega\|_\T\geq c_0|l|^{-A},\qquad
l\in\Z^\nu\setminus\{0\},\quad |l|\leq10N.
\]
Then, for any $x_0\in\T^\nu$, we have
\[
\#\{k=1,\ldots,N:x_0+k\omega\in S\pmod1\}<N^{1-\delta},
\]
where $0<\delta=\delta(c_0,A,\nu)<1$. The same bound holds for any interval of $N$ consecutive iterates.
\end{lemma}

\begin{remark}
The notation $\delta=\delta(c_0,A)$ is maintained throughout the article, with $\nu=2$. We decrease $\delta$ if necessary so that $0<\delta<10^{-3}$ and the degree conditions in Lemma~\ref{counting} hold at the scales used below.
\end{remark}

\subsection{Inductive step}
Let $\delta$ be as in Lemma~\ref{counting}, and fix
\[
1>b>1-\delta^2,\qquad\kappa=\delta^2/25.
\]
We use the finite Diophantine condition
\begin{equation}\label{finiteDC}
\|l\cdot\omega\|_\T\geq c_0|l|^{-A},\qquad
l\in\Z^2\setminus\{0\},\quad |l|\leq100N.
\end{equation}

\begin{lemma}[Inductive lemma]\label{induction}
Let $N_0,R_0,C_1$ be as in Lemma~\ref{initialhigh}, with $C_1$ enlarged if necessary. There are $\bar N_0\geq N_0$ and $c_4>0$, depending only on $v,w,c_0,A,b$, such that the following holds. Fix
\[
\lambda\in\R,\qquad |E|\geq R_0,\qquad 1\leq\gamma\leq3/2,
\]
integers $\bar N\geq\bar N_0$ and $\bar N<N\leq\bar N^{1/\delta^2}$, and a frequency $\omega$ satisfying \eqref{finiteDC} at scale $N$.

Suppose that, for every integer $N_0\leq M\leq\bar N$, there is a semi-algebraic set $X_M\subset\T^2$ of degree at most $M^{C_1}$ such that
\[
\sup_{x\in\T}\mes_1\{y\in\T:(x,y)\in X_M\}<e^{-M^\kappa},
\]
and, for $(x,y)\notin X_M$,
\begin{equation}\label{indhypbounds}
\|\mathcal G_M\|<e^{M^b},\qquad
|\mathcal G_M(m,n)|<10e^{-\gamma|m-n|}\quad(|m-n|>M/10).
\end{equation}
Then there is a semi-algebraic set $X_N\subset\T^2$ of degree at most $N^{C_1}$ satisfying
\begin{equation}\label{indsection}
\sup_{x\in\T}\mes_1\{y\in\T:(x,y)\in X_N\}<e^{-N^\kappa}.
\end{equation}
For $(x,y)\notin X_N$,
\begin{equation}\label{indconclusion}
\|\mathcal G_N\|<e^{N^b},\qquad
|\mathcal G_N(m,n)|<10e^{-(\gamma-\bar N^{-c_4})|m-n|}
\quad(|m-n|>N/10).
\end{equation}
For fixed $\lambda,\bar N,\gamma$, the constructed sets are jointly semi-algebraic in $(E,\omega,x,y)$. The constants $\bar N_0,c_4,C_1$ are independent of $\lambda,E,\gamma$.
\end{lemma}
\begin{proof}
For fixed real $x$, the matrix $B$ from \eqref{workingmatrix} is holomorphic in $y$, and
\begin{equation}\label{balancedstrip}
\sup_{|\Im y|<h}\|B_J(x,y)\|\leq1+\|w\|_h+2|E|^{-1}.
\end{equation}
It is self-adjoint for real $y$, its hopping satisfies
$|B(m,n)|\leq e^{-20|m-n|}$ for $m\ne n$, and it is translation covariant since $d_{n+k}(x)=d_n(x+k\rho)$. The resolvent and coupling lemmas therefore apply, with uniform constants.

Choose $\bar N_0$ large enough that $\bar N_0^{\delta^6}\geq2N_0$. Denote
\[
M_0=[N^{\delta^6}],\qquad L_0=[N^{\delta^2/100}],
\]
where $[\cdot]$ denotes the integer part. Then $N_0\leq M_0,L_0\leq\bar N$. Denote $M=[N^{\delta/5}]$. Let $J\subset[-N,N]$ be an $M$-interval such that for all $k\in J$,
\begin{align}
\|\mathcal G_{\Lambda_{M_0}(k)}\|&<e^{M_0^b},\label{smallnorm}\\
|\mathcal G_{\Lambda_{M_0}(k)}(m,n)|&<10e^{-\gamma|m-n|},\qquad\text{for }|m-n|>M_0/10.\label{smalldecay}
\end{align}
To apply Lemma~\ref{gluing}, write $J=[a_J,b_J]\cap\Z$. Since $M\gg M_0$, at each $k\in J$ choose
\[
c(k)=\min\{\max\{k,a_J+M_0\},b_J-M_0\},\qquad
I_k=\Lambda_{M_0}(c(k)).
\]
Then $I_k\subset J$, $|I_k|=2M_0+1$, and $\dist(k,J\setminus I_k)\geq M_0+1$. The center $c(k)$ lies in $J$, so \eqref{smallnorm}--\eqref{smalldecay} apply to $I_k$. Lemma~\ref{gluing}, with local size $2M_0+1$ and global size $|J|=M$, gives
\begin{align}
\|\mathcal G_J\|&<e^{M_0^{b_+}},\label{localnorm}\\
|\mathcal G_J(m,n)|&<e^{-(\gamma-N^{-\delta^7})|m-n|},\qquad |m-n|>M/10,\label{localdecay}
\end{align}
Indeed, the norm bound is $C(2M_0+1)e^{(2M_0+1)^b}<e^{M_0^{b_+}}$, where $b<b_+<1$. The loss in the decay rate is at most
\[
C\frac{\log M_0}{M_0}
+C\frac{M_0}{M}=o(N^{-\delta^7}).
\]
The same estimates hold for $|J|\asymp M$, with distance cutoff $|J|/10$, provided \eqref{smallnorm}--\eqref{smalldecay} hold for every $k\in J$.

Let $N^{\delta/5}<L\leq2N+1$, and let $\Lambda'\subset[-N,N]$ be an $L$-interval. Note that
\[
\log M_0^{C_1}\ll\log L<\log e^{M_0^\kappa}.
\]
The section estimate in the induction hypothesis, Fubini's theorem and Lemma~\ref{counting} give at most $L^{1-\delta}$ values of $k\in\Lambda'$ such that at least one of \eqref{smallnorm} and \eqref{smalldecay} fails. Here a translation of the initial phase allows any interval $\Lambda'$, and \eqref{finiteDC} covers the finite frequency cutoff. Thus, there are at most $L^{1-\delta}$ disjoint $M$-intervals $J\subset\Lambda'$ such that \eqref{localdecay} fails.

Now we estimate $\|\mathcal G_{\Lambda'}\|$. From the above argument, there exists a subset $\Lambda''\subset\Lambda'$ (a union of intervals of size $O(M_0)$ around the bad sites) such that
\begin{equation}\label{removedsize}
|\Lambda''|\leq CM_0L^{1-\delta}.
\end{equation}
To ensure the local covering required by Lemma~\ref{gluing}, remove the $2M_0$-neighborhood of each bad center and then absorb every remaining component of size less than $2M_0+1$ into $\Lambda''$. This preserves \eqref{removedsize}: there are at most one more components than bad centers. Each retained component has size at least $2M_0+1$, all its centers are good, and distinct retained components are separated by more than $2M_0$. The endpoint adjustment above gives the required local intervals in $\Lambda'\setminus\Lambda''$. Lemma~\ref{gluing} gives
\begin{equation}\label{removedinverse}
\|\mathcal G_{\Lambda'\setminus\Lambda''}\|
\leq C(2M_0+1)e^{(2M_0+1)^b}<e^{M_0}.
\end{equation}

On the other hand, the induction hypothesis at scale $L_0$ gives, for each fixed $x$,
\begin{equation}\label{crudesection}
\begin{split}
&\mes_1\{y:\|\mathcal G_{\Lambda'}(x,y)\|>e^{L_0}\}\\
&\quad\leq\sum_{k\in\Lambda'}
\mes_1\{y:(x+k\rho,y+k\alpha)\in X_{L_0}\}\\
&\quad\leq Le^{-L_0^\kappa}\ll e^{-2M_0}.
\end{split}
\end{equation}
Indeed, if $(x+k\rho,y+k\alpha)\notin X_{L_0}$ for every $k\in\Lambda'$, the radius-$L_0$ intervals may be chosen inside $\Lambda'$ by the same endpoint adjustment. Lemma~\ref{gluing}, with local size $2L_0+1$, gives $\|\mathcal G_{\Lambda'}\|\leq C(2L_0+1)e^{(2L_0+1)^b}<e^{L_0}$. The section estimate is uniform in $x+k\rho$, and translation by $k\alpha$ preserves one-dimensional measure. The last comparison follows from $L_0^\kappa=N^{\delta^4/2500+o(1)}\gg M_0$.

Fix $x\in\T$ and a real center $y_c$. Choose a fixed $0<a<\min\{h/2,1/4\}$ and apply Proposition~\ref{cartan} to the one-variable matrix
\[
A_x(z)=B_{\Lambda'}(x,y_c+z).
\]
The phases $x+n\rho$ are fixed as $z$ varies. Condition $(\mathrm{i})$ follows from \eqref{balancedstrip}, with $B_1=3+\|w\|_h$ independent of $x,y_c,\lambda,E$ and $\Lambda'$. In that proposition take
\[
M_{\mathrm C}=\lceil CM_0L^{1-\delta}\rceil+1,
\qquad B_2=e^{M_0},\qquad B_3=e^{L_0},
\]
where $M_{\mathrm C}$ denotes its parameter $M$. Equations~\eqref{removedsize}--\eqref{removedinverse} verify condition $(\mathrm{ii})$ at every real $z$. By \eqref{crudesection}, the exceptional measure in condition $(\mathrm{iii})$ is at most
\[
Le^{-L_0^\kappa}
<10^{-3}a(1+B_1)^{-1}(1+e^{M_0})^{-1}
\]
for sufficiently large $N$.

The product of the deleted-coordinate bound and the logarithmic factor satisfies
\[
M_{\mathrm C}\log(M_{\mathrm C}+B_1+B_2+B_3)
\leq L^{1-\delta}N^{\delta^2/50},
\]
because $M_0L_0\leq N^{\delta^6+\delta^2/100}\ll N^{\delta^2/50}$. Moreover,
\[
M_{\mathrm C}\log(1+B_1+e^{M_0})
\ll L^{1-2\delta^2},
\]
since $M_0^2\leq L^{10\delta^5}$ and $\delta-2\delta^2-10\delta^5>0$. This verifies the required smallness condition for $\eps=e^{-L^{1-2\delta^2}}$.

Proposition~\ref{cartan} now gives
\[
\begin{split}
&\mes_1\{z\in[-a/2,a/2]:
\|\mathcal G_{\Lambda'}(x,y_c+z)\|>e^{L^{1-2\delta^2}}\}\\
&\quad\leq C\exp\left[-c
\frac{L^{1-2\delta^2}}{L^{1-\delta}N^{\delta^2/50}}\right]
\leq C\exp\left[-cL^{9\delta/10-2\delta^2}\right].
\end{split}
\]
Here $N^{\delta^2/50}\leq L^{\delta/10}$. Covering $\T$ by finitely many intervals $y_c+[-a/2,a/2]$ and using $9\delta/10-2\delta^2>2\delta/5$, we obtain
\begin{equation}\label{cartansection}
\sup_{x\in\T}\mes_1\{y\in\T:\|\mathcal G_{\Lambda'}(x,y)\|>e^{L^{1-2\delta^2}}\}
<e^{-L^{2\delta/5}}.
\end{equation}
Taking the union over all such intervals in \eqref{cartansection}, we obtain an exceptional set $\widetilde X_N\subset\T^2$ with
\begin{equation}\label{indrawmeasure}
\sup_{x\in\T}\mes_1\{y\in\T:(x,y)\in\widetilde X_N\}
<CN^2e^{-N^{2\delta^2/25}}<e^{-N^\kappa}.
\end{equation}
For $(x,y)\notin\widetilde X_N$, we have
\[
\|\mathcal G_{\Lambda'}\|\leq e^{L^{1-2\delta^2}}
\]
for all $L$-intervals $\Lambda'\subset[-N,N]$ with $N^{\delta/5}<L\leq2N+1$. The factor $CN^2$ counts all subintervals. 

We apply Lemma~\ref{coupling} to every interval $J\subset\Lambda_N$ with $L_J:=|J|\geq N/20$, using $\tau=\delta/4$ and exponent $1-2\delta^2$. The norm estimates are available on every subinterval required by that lemma, since
\[
L_J^{\delta/4}>N^{\delta/5}
\]
for sufficiently large $N$. The gluing and orbit-counting arguments above also apply to local intervals of size $\asymp L_J^{\delta/4}$. Their bad-interval condition follows from
\[
CL_J^{\delta/4}L_J^{1-\delta}<L_J^{1-2\delta^2}.
\]
Consequently, for some $c=c(\delta)>0$, these restrictions have decay rate at least $\gamma-CN^{-c}$ at distances greater than $L_J/10$. In particular,
\begin{equation}\label{rawindnorm}
\|\mathcal G_N\|\leq e^{(2N+1)^{1-2\delta^2}}.
\end{equation}
To obtain the cutoff $N/10$ for the radius-$N$ interval, suppose
$N/10<d=n-m\leq(2N+1)/10$ and take
$J=[\max(-N,m-d),n-1]\cap\Z$. Then $d\leq|J|\leq2d$, and every point of $\Lambda_N\setminus J$ is at distance at least $d$ from $m$. In the resolvent identity \eqref{glueresolvent}, split the local sum according to $|m-m_1|\leq|J|/10$ or $|m-m_1|>|J|/10$. The first part uses the local norm bound and $20>10\gamma$; the second uses the preceding local decay. Together with \eqref{rawindnorm}, this gives
\[
|\mathcal G_N(m,n)|\leq
CN^2e^{2(2N+1)^{1-2\delta^2}}
e^{-(\gamma-CN^{-c})d}.
\]
Since $d>N/10$, the prefactor costs $O(N^{-2\delta^2}+\log N/N)$ in the decay rate. For larger $d$ the coupling estimate on $\Lambda_N$ applies directly, and $n<m$ follows by symmetry. Thus, after decreasing $c_3(\delta)>0$,
\begin{equation}\label{rawinddecay}
|\mathcal G_N(m,n)|<e^{-(\gamma-\bar N^{-c_3(\delta)})|m-n|},
\qquad |m-n|>N/10.
\end{equation}

We now obtain a semi-algebraic substitution of $\widetilde X_N$. Use the relative polynomial approximation \eqref{polynomial}--\eqref{relativepoly} at scale $N$ and the matrix $\mathcal G_N^{(1)}$ in \eqref{polyweightedgreen}. Choose $0<c_4<c_3(\delta)$ and define
\begin{equation}\label{indpolygood}
\begin{split}
X_N^c=\bigl\{(x,y):\ &\det\HH_N^{(1)}\ne0,\quad
\|\mathcal G_N^{(1)}\|_{\rm HS}<2Ne^{(2N+1)^{1-2\delta^2}},\\
&|\mathcal G_N^{(1)}(m,n)|<2e^{-(\gamma-\bar N^{-c_4})|m-n|}
\text{ for all }|m-n|>N/10\bigr\}.
\end{split}
\end{equation}
The polynomial description \eqref{weightedpolysquares} applies without change and gives degree at most $N^{C_1}$, jointly in $(E,\omega,x,y)$ for fixed $\lambda$.

Using \eqref{balancedpolyerror}--\eqref{polyweightedgreen}, we repeat the two perturbation steps in the proof of Lemma~\ref{initialhigh}. By \eqref{rawindnorm}--\eqref{rawinddecay}, the first step gives $X_N\subset\widetilde X_N$: the Hilbert--Schmidt factor is at most $\sqrt{2N+1}$, and $c_4<c_3$ gives the strict margin in the decay estimate. Thus \eqref{indrawmeasure} proves \eqref{indsection}. The reverse step proves \eqref{indconclusion}. The smallness condition follows from
\[
e^{-N^3}(2N+1)\bigl(3Ne^{(2N+1)^{1-2\delta^2}}\bigr)^2e^{2N/5}\longrightarrow0,
\]
and $6Ne^{(2N+1)^{1-2\delta^2}}<e^{N^b}$. All perturbation constants are independent of $\lambda,E$.

\end{proof}

\begin{proposition}[Large deviations at high energy]\label{ldthigh}
For any real analytic $v$, possibly constant, there are $E_0=E_0(v,w,c_0,A)>1$, $N_0$ and $C_1$ such that, for every
\[
\lambda\in\R,\qquad |E|\geq E_0,\qquad N\geq N_0,
\]
and every frequency satisfying \eqref{finiteDC}, there is a semi-algebraic set $X_N(\lambda,E,\omega)\subset\T^2$ of degree at most $N^{C_1}$ satisfying
\begin{equation}\label{ldtsection}
\sup_{x\in\T}\mes_1\{y\in\T:(x,y)\in X_N(\lambda,E,\omega)\}<e^{-N^\kappa}.
\end{equation}
For all $(x,y)\notin X_N$,
\begin{align}
\|G_N(\lambda,E)\|&\leq e^{N^b},\label{ldtnorm}\\
|G_N(\lambda,E)(m,n)|&\leq10e^{-|m-n|/2},\qquad |m-n|>N/10.\label{ldtdecay}
\end{align}
For fixed $\lambda$, the sets are jointly semi-algebraic in $(E,\omega,x,y)$. The constants $E_0,N_0,C_1,b,\kappa$ are independent of $\lambda,E$.
\end{proposition}
\begin{proof}
Fix $b=1-\delta^2/2$ and $\kappa=\delta^2/25$. Enlarge the starting scale $\bar N_0$ in Lemma~\ref{induction} so that, with
\[
\bar N_{j+1}=\lfloor\bar N_j^{1/\delta^2}\rfloor,
\qquad j\geq0,
\]
we have
\[
\sum_{j\geq0}\bar N_j^{-c_4}<\tfrac12.
\]
Choose
\[
E_0\geq\max\{R_0,\exp(\bar N_0^{2\kappa})\}.
\]
For every $\lambda\in\R$ and $|E|\geq E_0$, Lemma~\ref{initialhigh} gives the section estimate and the bounds for $\mathcal G_M$ at every $N_0\leq M\leq\bar N_0$, with decay rate $3/2$.

Set
\[
\gamma^{(j)}=\frac32-\sum_{i=0}^{j-1}\bar N_i^{-c_4},\qquad j\geq0,
\]
where the sum is empty for $j=0$. Suppose the estimates hold through $\bar N_j$ with decay rate $\gamma^{(j)}$. Lemma~\ref{induction}, with $\bar N=\bar N_j$, gives them for every $\bar N_j<M\leq\bar N_{j+1}$, with decay rate $\gamma^{(j+1)}$. The earlier estimates remain valid with this smaller rate. Since $1<\gamma^{(j)}\leq3/2$ for every $j$, the lemma can be applied at every stage, and these intervals cover all $M\geq N_0$.

For a prescribed final scale $N$, the finite Diophantine condition \eqref{finiteDC} implies the conditions at every smaller scale used in this argument. In the last stage we extend only as far as $N$. Thus \eqref{ldtsection} holds at every scale in the proposition, and outside $X_N$,
\[
\|\mathcal G_N\|<e^{N^b},\qquad
|\mathcal G_N(m,n)|<10e^{-|m-n|}\quad(|m-n|>N/10).
\]
Equation~\eqref{balancedtransfer} gives \eqref{ldtnorm}--\eqref{ldtdecay}. Each $X_N$ has the polynomial description constructed at its own scale, so its degree is at most $N^{C_1}$ and the sets are jointly semi-algebraic in $(E,\omega,x,y)$ for fixed $\lambda$. All scales and decay rates above are independent of $\lambda,E$; hence so are $E_0,N_0,C_1,b,\kappa$.
\end{proof}

\subsection{Large coupling}
We normalize by $\lambda$ or $E$ according to which has larger absolute value. The two-frequency exclusion requires a bound on the total phase measure.

\begin{proposition}[Large deviations at large coupling]\label{ldtlarge}
Suppose $v$ is nonconstant. There are $\lambda_0=\lambda_0(v,w,c_0,A)>1$, $N_0$ and $C_1$ such that, for every
\[
|\lambda|>\lambda_0,\qquad E\in\R,\qquad N\geq N_0,
\]
and every frequency satisfying \eqref{finiteDC}, there is a semi-algebraic set $Y_N(\lambda,E,\omega)\subset\T^2$ of degree at most $N^{C_1}$ satisfying
\begin{equation}\label{largephasebound}
\mes_2Y_N(\lambda,E,\omega)<e^{-N^\kappa}.
\end{equation}
Outside $Y_N$, the Green's function satisfies \eqref{ldtnorm}--\eqref{ldtdecay}. For fixed $\lambda$, the sets are jointly semi-algebraic in $(E,\omega,x,y)$. The constants are independent of $\lambda,E$.
\end{proposition}
\begin{proof}
Denote
\begin{equation}\label{largenormalizations}
\widetilde H=
\begin{cases}
\lambda^{-1}\Delta+V-(E/\lambda)W,& |E|\leq|\lambda|,\\
E^{-1}\Delta+(\lambda/E)V-W,& |E|>|\lambda|,
\end{cases}
\qquad \widetilde G_J=(\widetilde H_J)^{-1}.
\end{equation}
In both cases the diagonal is real analytic with norm at most $\|v\|_h+\|w\|_h$, and the hopping is bounded by $e^{-20|m-n|}$ for $m\ne n$ once $\lambda_0\geq e^{20}$. Moreover,
\begin{equation}\label{largetransfer}
G_J=
\begin{cases}
\lambda^{-1}\widetilde G_J,& |E|\leq|\lambda|,\\
E^{-1}\widetilde G_J,& |E|>|\lambda|.
\end{cases}
\end{equation}
In particular, $E=0$ belongs to the first case, where $\widetilde H=\lambda^{-1}\Delta+V$.

\emph{Initial step.}
Let $N_0\leq N\leq(\log|\lambda|)^{1/(2\kappa)}$ and set $\eta_N=e^{-N^{4\kappa/3}}$.
For $|E|\leq|\lambda|$, let
\[
\widetilde Y_N=
\left\{(x,y):\left|v(x+n\rho)-\frac E\lambda w(y+n\alpha)\right|<\eta_N
\text{ for some }|n|\leq N\right\}.
\]
Applying the scalar-level estimate \eqref{scalarlevel} to the nonconstant function $v$, for every fixed $y$, gives
\[
\sup_{y\in\T}\mes_1\{x:(x,y)\in\widetilde Y_N\}
\leq C(2N+1)\eta_N^{c(v)}<e^{-2N^\kappa}.
\]
Write $D_N=V_N-(E/\lambda)W_N$. For $(x,y)\notin\widetilde Y_N$,
\[
|D_N(n,n)|\geq\eta_N>|\lambda|^{-1/2},\qquad
2|\lambda|^{-1}\eta_N^{-1}<\tfrac12,
\]
so the Neumann expansion and nearest-neighbor path estimates in \eqref{neumannhighseries}--\eqref{rawdecayhigh} give
\[
\begin{split}
\|\widetilde G_N\|&\leq2e^{N^{4\kappa/3}},\\
|\widetilde G_N(m,n)|&\leq2e^{N^{4\kappa/3}}e^{-2|m-n|}
<e^{-7|m-n|/4},\qquad |m-n|>N/10.
\end{split}
\]
For $|E|>|\lambda|$, replace the diagonal by $(\lambda/E)V_N-W_N$ and the hopping by $E^{-1}\Delta_N$. The scalar-level estimate for $w$ now gives the same bound on every $y$-section, uniformly in $x$. Since $|\lambda/E|<1$ and $\log|E|>\log|\lambda|$, all the preceding Neumann-series and path estimates remain valid.

Choose real polynomials $v_1,w_1$ of degree less than $N^4$ such that
\[
\|v-v_1\|_{[0,1]}+\|w-w_1\|_{[0,1]}<\epsilon_N,
\qquad \epsilon_N=e^{-N^3}.
\]
Let $\widetilde H^{(1)}$ be the corresponding polynomial replacement in \eqref{largenormalizations}, with the arguments evaluated modulo one, and write $\widetilde G_N^{(1)}=(\widetilde H_N^{(1)})^{-1}$. In the respective cases,
\begin{equation}\label{largepolyerror}
\begin{split}
\left|(v-v_1)-\frac E\lambda(w-w_1)\right|&<\epsilon_N,
\qquad |E|\leq|\lambda|,\\
\left|\frac\lambda E(v-v_1)-(w-w_1)\right|&<\epsilon_N,
\qquad |E|>|\lambda|.
\end{split}
\end{equation}
The two functions in each line are evaluated in their respective phase variables. Thus absolute polynomial approximation suffices in both cases.

Define $Y_N^c$ by nonsingularity of $\widetilde H_N^{(1)}$ and
\[
\|\widetilde G_N^{(1)}\|_{\rm HS}<4(2N+1)e^{N^{4\kappa/3}},\qquad
|\widetilde G_N^{(1)}(m,n)|<2e^{-3|m-n|/2},\quad |m-n|>N/10.
\]
By Lemma~\ref{perturb}, first in the forward direction and then in the reverse direction, $Y_N\subset\widetilde Y_N$, and outside $Y_N$,
\[
\|\widetilde G_N\|<e^{N^b},\qquad
|\widetilde G_N(m,n)|<10e^{-3|m-n|/2},\quad |m-n|>N/10.
\]
The smallness check is the one at the end of the proof of Lemma~\ref{initialhigh}, using \eqref{largepolyerror}. The phase partition and Cramer's rule give degree at most $N^{C_1}$. We retain the $x$-section estimate when $|E|\leq|\lambda|$ and the $y$-section estimate when $|E|>|\lambda|$.

\emph{Inductive step.}
Let $\bar N<N\leq\bar N^{1/\delta^2}$ and assume that the preceding section bound and
\[
\|\widetilde G_M\|<e^{M^b},\qquad
|\widetilde G_M(m,n)|<10e^{-\gamma|m-n|},\quad |m-n|>M/10
\]
hold outside $Y_M$ for every $N_0\leq M\leq\bar N$, where $1\leq\gamma\leq3/2$. We use the same scales
\[
M_0=[N^{\delta^6}],\qquad L_0=[N^{\delta^2/100}],\qquad M=[N^{\delta/5}]
\]
as in Lemma~\ref{induction}. The resolvent iteration and orbit counting in that proof apply to $\widetilde H$: it is self-adjoint at real phases, translation covariant, and has the same hopping bound. In particular, on an $L$-interval $\Lambda'$ with $L>N^{\delta/5}$, at most $L^{1-\delta}$ centers are bad. Removing their $2M_0$-neighborhoods and the short remaining components gives
\[
|\Lambda''|\leq CM_0L^{1-\delta},\qquad
\|\widetilde G_{\Lambda'\setminus\Lambda''}\|<e^{M_0}.
\]

For $|E|\leq|\lambda|$, fix $y$ and apply the one-variable Cartan estimate in $x$. The analytic norm is bounded by $\|v\|_h+\|w\|_\T+2|\lambda|^{-1}$. The estimate at scale $L_0$ and the endpoint-adjusted resolvent covering give
\[
\sup_y\mes_1\{x:\|\widetilde G_{\Lambda'}(x,y)\|>e^{L_0}\}
\leq Le^{-L_0^\kappa}\ll e^{-2M_0}.
\]
For $|E|>|\lambda|$, fix $x$ and apply Cartan in $y$ instead; the analytic norm is bounded by $\|v\|_\T+\|w\|_h+2|E|^{-1}$, and the same coarse estimate holds for the $y$-sections. In both cases take $B_1=3+\|v\|_h+\|w\|_h$ and
\[
M_{\mathrm C}=\lceil CM_0L^{1-\delta}\rceil+1,
\qquad B_2=e^{M_0},\qquad B_3=e^{L_0}.
\]
The Cartan smallness conditions and exponent comparisons in the proof of Lemma~\ref{induction} are unchanged. They give
\[
\begin{aligned}
\sup_y\mes_1\{x:\|\widetilde G_{\Lambda'}(x,y)\|>e^{L^{1-2\delta^2}}\}
&<e^{-L^{2\delta/5}},\qquad |E|\leq|\lambda|,\\
\sup_x\mes_1\{y:\|\widetilde G_{\Lambda'}(x,y)\|>e^{L^{1-2\delta^2}}\}
&<e^{-L^{2\delta/5}},\qquad |E|>|\lambda|.
\end{aligned}
\]
Taking the union over all subintervals and then applying Lemma~\ref{coupling}, including the radius-to-length adjustment used in \eqref{rawinddecay}, gives an exceptional set $\widetilde Y_N$ with the appropriate section measure less than $e^{-N^\kappa}$, outside which
\[
\|\widetilde G_N\|\leq e^{(2N+1)^{1-2\delta^2}},\qquad
|\widetilde G_N(m,n)|<e^{-(\gamma-\bar N^{-c_3})|m-n|}
,\quad |m-n|>N/10.
\]

Using the same absolute polynomial approximation at scale $N$, choose $0<c_4<c_3$ and define $Y_N^c$ by nonsingularity and
\[
\begin{split}
\|\widetilde G_N^{(1)}\|_{\rm HS}&<2Ne^{(2N+1)^{1-2\delta^2}},\\
|\widetilde G_N^{(1)}(m,n)|&<2e^{-(\gamma-\bar N^{-c_4})|m-n|}
,\quad |m-n|>N/10.
\end{split}
\]
By \eqref{largepolyerror} and the two applications of Lemma~\ref{perturb}, $Y_N\subset\widetilde Y_N$, while outside $Y_N$,
\[
\|\widetilde G_N\|<e^{N^b},\qquad
|\widetilde G_N(m,n)|<10e^{-(\gamma-\bar N^{-c_4})|m-n|}
,\quad |m-n|>N/10.
\]
By \eqref{largepolyerror}, the perturbation smallness condition is the same as at the end of Lemma~\ref{induction}. Thus the induction preserves the section estimate in each case, and Fubini's theorem gives \eqref{largephasebound}.

Choose $\bar N_0$ as in Proposition~\ref{ldthigh}, so that the induction applies and the total decay loss along its scale sequence is less than $1/2$. Enlarge $\lambda_0$ to satisfy $\lambda_0\geq\exp(\bar N_0^{2\kappa})$. The initial estimates then hold through $\bar N_0$ with decay rate $3/2$ in both cases. The same iteration reaches every target scale under \eqref{finiteDC}, with decay rate greater than $1$. Equation~\eqref{largetransfer} gives \eqref{ldtnorm}--\eqref{ldtdecay} for $G_N$.

Finally, the two cases above are distinguished by the polynomial inequalities $E^2\leq\lambda^2$ and $E^2>\lambda^2$. In each case, partition according to the integer parts of the shifted phases and use Cramer's rule. Clearing the squared determinant and scalar denominators gives the joint semi-algebraic description in $(E,\omega,x,y)$ for fixed $\lambda$. The scalar denominator is $\lambda\ne0$ in the first case and $E\ne0$ in the second, so this construction includes $E=0$ without any additional restriction. Enlarge $N_0,C_1$ to cover both cases.
\end{proof}

\section{Exponential decay of generalized eigenvectors}
We first fix the potential frequency $\rho$ and prove high-energy exponential decay for generalized eigenvectors for almost every weight frequency $\alpha$ in the corresponding Diophantine slice.

\subsection{Preparations}
\begin{lemma}\label{projection}
Let $S\subset\R^n$ be a semi-algebraic set of degree $B$. Then any projection of $S$ is semi-algebraic of degree at most $B^{C_0}$, where $C_0=C_0(n)$.
\end{lemma}

\begin{lemma}[Semi-algebraic transversality, {\cite[Lemma~9.9]{Bou05}}]\label{transversality}
Suppose that $S\subset[0,1]^{2n}$ is a semi-algebraic set of degree $B$ with $\mes_{2n}S<\eta$ and $\log B\ll\log\eta^{-1}$. Denote by $(\omega,z)\in[0,1]^n\times[0,1]^n$ the product variable. Moreover, denote by $\{e_j:0\leq j\leq n-1\}$ the $\omega$-coordinate vectors. Fix $\eps>\eta^{1/(2n)}$. Then there exists a decomposition $S=S_1\cup S_2$, such that $S_1$ satisfies
\[
\mes_n(\Proj_\omega S_1)\leq B^{C_0}\eps.
\]
Furthermore, $S_2$ satisfies the transversality property
\[
\mes_n(S_2\cap L)\leq B^{C_0}\eps^{-1}\eta^{1/(2n)}
\]
for any affine $n$-plane $L$ whose direction space satisfies
\[
\max_{0\leq j\leq n-1}\|\Proj_L e_j\|<\eps/100.
\]
Here $C_0=C_0(n)$, and $\Proj_L$ denotes orthogonal projection onto the direction of $L$.
\end{lemma}

\subsection{Semi-algebraic estimates at a fixed potential frequency}\label{sliceestimates}
All torus coordinates are taken modulo one. For semi-algebraic descriptions we use unit-interval representatives and split at their endpoints. The standard component and projection bounds are used as in \cite[Chapter~9]{Bou05}.

\begin{lemma}[Fixed-first-frequency avoidance]\label{fixedrhoavoid}
Fix $\rho\in\T$ satisfying
\begin{equation}\label{fixedrhoarith}
\|k\rho\|_\T\geq c_\rho|k|^{-A}\quad(k\ne0),\qquad c_\rho>0,\quad A\geq1.
\end{equation}
There are absolute constants $C,C_*$ and a sufficiently large absolute $R_0$ with the following property. Let $S\subset\T_\alpha\times\T_x\times\T_y$ be semi-algebraic of degree at most $B\ge2$. Suppose
\begin{equation}\label{fixedrhoarea}
 \sup_{x\in\T}\mes_2(S_x)\le\eta,
 \qquad S_x=\{(\alpha,y):(\alpha,x,y)\in S\}.
\end{equation}
Let $R$ be an integer satisfying
\begin{equation}\label{fixedrhoscales}
 R\ge \max(R_0,B^{C_*}),\qquad 0<\eta\le R^{-20}.
\end{equation}
Then, uniformly in $x_0,y_0\in\T$,
\begin{equation}\label{fixedrhobound}
\begin{aligned}
 &\mes_1\{\alpha:\text{ for some integer }R\le n<2R,\\[-2mm]
 &\hspace{37mm}(\alpha,x_0+n\rho,y_0+n\alpha)\in S\}\\
 &\quad\le C(c_\rho,A)B^C R^{-\sigma}
       +CB^C\bigl(R\eta^{1/4}+R^2\eta^{1/2}\bigr),
 \qquad \sigma=\frac1{A+2}.
\end{aligned}
\end{equation}
The same conclusion holds for negative integers with $R\le |n|<2R$. Here and below, $C(c_\rho,A)$ denotes a positive constant depending only on $c_\rho$ and $A$, whose value may change from line to line. In particular, these constants are uniform in $\rho$ satisfying \eqref{fixedrhoarith} with the stated $c_\rho,A$.
\end{lemma}
The proof of Lemma~\ref{fixedrhoavoid} separates positions $n$ according to whether $S_{x_0+n\rho}$ contains a long segment of slope $n$. Lemma~\ref{sliceneighborhood} provides the neighborhood and grid bounds used in the direction estimate of Lemma~\ref{slicedirections} and the orbit count of Lemma~\ref{sliceslowcount}. Lemma~\ref{sliceoneposition} gives the measure estimate at a single position.

\begin{lemma}[Planar semi-algebraic neighborhoods]\label{sliceneighborhood}
Let $U\subset[0,1]^2$ be semi-algebraic of degree at most $D\ge2$, and write $U_s=\{z\in\R^2:\dist(z,U)<s\}$. For $0<s<1$,
\begin{equation}\label{slicetubearea}
 \mes_2(U_s)\le \mes_2(U)+CD^Cs.
\end{equation}
The same estimate holds in a fixed translated unit square. In particular, the number of squares of a square grid of mesh $h$ meeting $U$ is at most
\begin{equation}\label{slicegridnumber}
 C\bigl(\mes_2(U)h^{-2}+D^Ch^{-1}\bigr).
\end{equation}
\end{lemma}
\begin{proof}
The boundary $\partial U$ is semi-algebraic of degree at most $D^C$ and dimension at most one. By \cite[Theorem~9.4, pp.~49--50]{Bou05}, with $r=1$, it is covered by at most $CD^C$ Lipschitz images of unit intervals and points, with uniformly bounded Lipschitz constants. Subdivide each parameter interval into pieces of length at most $s$. Its image is then covered by $O(s^{-1})$ disks of radius $Cs$, so its $s$-neighborhood has area $O(s)$. Summing over the images gives
\[
\mes_2((\partial U)_s)\leq CD^Cs.
\]
Since $U_s\setminus U\subset(\partial U)_s$, this proves \eqref{slicetubearea}.

Every grid square meeting $U$ is contained in $U_{2h}$. These squares have disjoint interiors and area $h^2$, so \eqref{slicetubearea} gives \eqref{slicegridnumber} when $2h<1$. For $h\geq1/2$, a bounded number of squares suffices.
\end{proof}

\begin{lemma}[Directions of contained segments]\label{slicedirections}
Suppose $U\subset[0,1]^2$ is semi-algebraic of degree at most $B$ and $\mes_2(U)\le\eta$. Set
\[
 \ell=\eta^{1/4},\qquad d=\eta^2.
\]
If $CB^C\eta\le1$ for a sufficiently large absolute power and constant, the set $\Theta(U)$ of unoriented directions in which $U$ contains a line segment of length $\ell$ satisfies
\begin{equation}\label{slicedirectionbound}
 \mes_1(\Theta(U))\le C\eta^{1/2}\log(1/\eta).
\end{equation}
\end{lemma}
\begin{proof}
We use an elementary planar Kakeya-type argument. Take angular coordinates modulo $\pi$ for unoriented directions. Split the angular circle into a bounded number of intervals of length at most $\pi/4$, and work on one interval. Choose a maximal $d/\ell$-separated family
\[
\theta_1<\cdots<\theta_M
\]
of directions in $\Theta(U)$ in this interval. For each $\theta_j$, choose a contained segment of length $\ell$ and thicken it to a rectangle $T_j$ of width $d$. Then $T_j\subset U_{Cd}$. Lemma~\ref{sliceneighborhood} and $CB^C\eta\leq1$ give
\[
\left|\bigcup_{j=1}^M T_j\right|
\leq \mes_2(U_{Cd})\leq\eta+CB^Cd\leq C\eta.
\]

For $i\ne j$, the intersection of two such rectangles satisfies
\[
|T_i\cap T_j|\leq C\min\left(d\ell,\frac{d^2}{|\theta_i-\theta_j|}\right).
\]
The case $M=0$ is immediate, so assume $M\geq1$. The separation gives
$|\theta_i-\theta_j|\geq|i-j|d/\ell$, and hence
$|T_i\cap T_j|\leq Cd\ell/|i-j|$ for $i\ne j$.
Put $F=\sum_{j=1}^M\mathbf1_{T_j}$. Expanding the square and grouping
the pairs according to $k=j-i$, we obtain
\begin{align*}
\int_{\R^2}F^2
&=\sum_{i,j=1}^M|T_i\cap T_j|\\
&=Md\ell+2\sum_{1\leq i<j\leq M}|T_i\cap T_j|\\
&\leq Md\ell+Cd\ell\sum_{k=1}^{M-1}\frac{M-k}{k}\\
&\leq CMd\ell\log(2M).
\end{align*}
Here $|T_j|=d\ell$, there are $M-k$ pairs with $j-i=k$, and
$\sum_{k=1}^{M-1}k^{-1}\leq C\log(2M)$, with the sum empty when $M=1$.
Since the directions lie in an interval of length at most $\pi/4$,
their separation also gives $(M-1)d/\ell\leq\pi/4$. Thus
$M\leq C\ell/d$, and the preceding bound implies
\[
\int_{\R^2}F^2\leq CMd\ell\log(2\ell/d).
\]

The function $F$ vanishes outside $\bigcup_{j=1}^M T_j$ and has
integral $\int_{\R^2}F=Md\ell$. Applying Cauchy--Schwarz on this union gives
\begin{align*}
(Md\ell)^2
&=\left(\int_{\bigcup_{j=1}^M T_j}F\right)^2\\
&\leq\left|\bigcup_{j=1}^M T_j\right|
\int_{\bigcup_{j=1}^M T_j}F^2\\
&\leq C\left|\bigcup_{j=1}^M T_j\right|
Md\ell\log(2\ell/d).
\end{align*}
Dividing by $CMd\ell\log(2\ell/d)$ therefore yields
\[
\left|\bigcup_{j=1}^M T_j\right|
\geq\frac{cMd\ell}{\log(2\ell/d)}.
\]
Combining the two area bounds yields
$Md/\ell\leq C\eta\ell^{-2}\log(2\ell/d)$.
By maximality, every direction under consideration is within $d/\ell$ of one of the $\theta_j$. Its total measure is thus at most $CMd/\ell$. Summing over the angular intervals and substituting $\ell=\eta^{1/4}$ and $d=\eta^2$ proves \eqref{slicedirectionbound}.
\end{proof}

\begin{lemma}[Counting for a fixed rotation with a slow coordinate]\label{sliceslowcount}
Assume \eqref{fixedrhoarith}. Let $Q\subset[1,2]\times\T$ be semi-algebraic of degree at most $D\ge2$, and suppose $\mes_2(Q)\le R^{-2\sigma}$, where $\sigma=1/(A+2)$. Then
\begin{equation}\label{sliceslowbound}
 \#\{R\le n<2R:(n/R,x_0+n\rho)\in Q\}
 \le C(c_\rho,A)D^CR^{1-\sigma}
\end{equation}
uniformly in $x_0$.
\end{lemma}
\begin{proof}
We first estimate visits of the rotation to an interval. Let
$I=\{a,a+1,\ldots,a+L-1\}$, where $a\in\Z$ and $L\geq1$, and let
$J\subset\T$ be an interval. For every integer $H\geq1$, the
Erd\H{o}s--Tur\'an inequality applied to the points $x_0+n\rho$,
$n\in I$, gives
\[
\begin{split}
&\left|\#\{n\in I:x_0+n\rho\in J\}-L|J|\right|\\
&\quad\leq C\left(\frac L{H+1}
+\sum_{k=1}^H\frac1k
\left|\sum_{n\in I}e^{2\pi i k(x_0+n\rho)}\right|\right).
\end{split}
\]
The geometric-series formula gives, for $k\geq1$,
\[
\left|\sum_{n\in I}e^{2\pi i k(x_0+n\rho)}\right|
=\left|\frac{1-e^{2\pi i kL\rho}}{1-e^{2\pi i k\rho}}\right|
\leq\min\left\{L,\frac1{2\|k\rho\|_\T}\right\}.
\]
Using \eqref{fixedrhoarith} and $A\geq1$, we therefore have
\[
\sum_{k=1}^H\frac1k
\left|\sum_{n\in I}e^{2\pi i k(x_0+n\rho)}\right|
\leq\frac1{2c_\rho}\sum_{k=1}^H k^{A-1}
\leq\frac{H^A}{2c_\rho}.
\]
Consequently,
\[
\left|\#\{n\in I:x_0+n\rho\in J\}-L|J|\right|
\leq C(c_\rho,A)\left(\frac L{H+1}+H^A\right).
\]
Set $\nu=1/(A+1)$ and choose $H=\lfloor L^\nu\rfloor\geq1$.
Then $H+1>L^\nu$ and $H^A\leq L^{A\nu}=L^{1-\nu}$, so both
terms on the right are bounded by $L^{1-\nu}$. Thus
\begin{equation}\label{slicediscrepancy}
\#\{n\in I:x_0+n\rho\in J\}
\leq L|J|+C(c_\rho,A)L^{1-\nu},
\qquad \nu=\frac1{A+1}.
\end{equation}
The estimate is uniform in the initial phase $x_0$, the starting
integer $a$ and the interval $J$.

Now take a square grid of mesh $h\asymp R^{-\sigma}$ with $h^2>\mes_2(Q)$. By \eqref{slicegridnumber}, after translating the first coordinate, at most $CD^Ch^{-1}$ squares meet $Q$. For one such square $I_1\times J$, the integers satisfying $R\leq n<2R$ and $n/R\in I_1$ form a consecutive block of length $L\leq Rh+2$, while $|J|\leq h$. If the block is nonempty, \eqref{slicediscrepancy} bounds its visits to the square by
\[
\begin{split}
L|J|+C(c_\rho,A)L^{1-\nu}
&\leq (Rh+2)h+C(c_\rho,A)(Rh+2)^{1-\nu}\\
&\leq C(c_\rho,A)\bigl(Rh^2+(Rh)^{1-\nu}+1\bigr).
\end{split}
\]
An empty block contributes no visits.
Summing over the squares gives
\[
C(c_\rho,A)D^C\bigl(Rh+R^{1-\nu}h^{-\nu}+h^{-1}\bigr).
\]
Since $h\asymp R^{-\sigma}$ and $\sigma=\nu/(1+\nu)$,
\[
Rh\asymp R^{1-\sigma},\qquad
R^{1-\nu}h^{-\nu}\asymp R^{1-\nu+\sigma\nu}=R^{1-\sigma}.
\]
Also $h^{-1}\asymp R^\sigma\leq R^{1-\sigma}$, since $\sigma\leq1/3$. This proves \eqref{sliceslowbound}. The bounds depend on $Q$ only through its degree, so its defining coefficients may depend on $\rho$ and $R$.
\end{proof}

The one-dimensional transversality estimate also gives the following bound at a single position.

\begin{lemma}[One position]\label{sliceoneposition}
Let $U\subset[0,1]^2$ be semi-algebraic of degree at most $B$ and area at most $\eta$, with the scale conditions \eqref{fixedrhoscales}. If $R\le n<2R$, then
\begin{equation}\label{sliceonebound}
 \mes_1\{\alpha:(\alpha,\{y_0+n\alpha\})\in U\}
 \le CB^C(R^{-1}+R\eta^{1/2}).
\end{equation}
\end{lemma}
\begin{proof}
Apply Lemma~\ref{transversality} in dimension one to $U$, with $\eps=200/R$. The scale assumptions imply $\eps>\eta^{1/2}$ and the required separation between degree and measure. We obtain $U=U_1\cup U_2$ with
\[
\mes_1(\Proj_\alpha U_1)\leq CB^C/R.
\]
For a line $L$ of slope $n$, the projection of the $\alpha$-coordinate vector onto its direction has length
\[
(1+n^2)^{-1/2}\leq R^{-1}<\eps/100.
\]
Thus the second part of that lemma gives
\[
\mes_1(U_2\cap L)\leq CB^CR\eta^{1/2},
\]
where $\mes_1$ is arclength on $L$, as in Lemma~\ref{transversality}.

The graph $\alpha\mapsto(\alpha,\{y_0+n\alpha\})$ has at most $n+2$ affine pieces, all of slope $n$. On each piece its arclength element is $ds=\sqrt{1+n^2}\,d\alpha$. The measure of the corresponding set of $\alpha$ is therefore its arclength divided by $\sqrt{1+n^2}$. Summing over the pieces, we obtain
\[
\begin{split}
&\mes_1\{\alpha:(\alpha,\{y_0+n\alpha\})\in U_2\}\\
&\quad\leq\frac{n+2}{\sqrt{1+n^2}}\,CB^CR\eta^{1/2}
\leq CB^CR\eta^{1/2}.
\end{split}
\]
The intersections with $U_1$ project into the single set $\Proj_\alpha U_1$, whose measure is already bounded above. Adding the two contributions proves \eqref{sliceonebound}.
\end{proof}

\begin{proof}[Proof of Lemma~\ref{fixedrhoavoid}]
Put $\ell=\eta^{1/4}$ and $x_n=\{x_0+n\rho\}$. For $R\leq n<2R$, denote
\[
E_n=\{\alpha:(\alpha,x_n,\{y_0+n\alpha\})\in S\}.
\]
We estimate the union of these sets by separating positions according to the lengths of the segments contained in $S_{x_n}$.

\emph{Step 1: the set of long segments.}
Define
\begin{equation}\label{slicelongset}
\begin{split}
\mathfrak D_R=\bigl\{(t,x)\in[1,2]\times[0,1]:\ &\exists a,b,\quad
0\leq a\leq1-\ell/(Rt),\quad0\leq b\leq1-\ell,\\
&(a+s/(Rt),x,b+s)\in S\quad\text{for every }s\in[0,\ell]\bigr\}.
\end{split}
\end{equation}
Thus $(t,x)\in\mathfrak D_R$ means that $S_x$ contains a segment of slope $Rt$ and vertical length $\ell$, entirely within its fundamental square. The variable $t$ allows us to describe the slopes near $R$ in a fixed interval.

The quantified formula has a fixed number of real variables. After multiplying by positive powers of $Rt$, effective quantifier elimination gives
\begin{equation}\label{slicelongdegree}
\deg\mathfrak D_R\leq B^C.
\end{equation}
Here $R$ and $\ell$ enter as coefficients, so the bound is uniform in both parameters.

\emph{Step 2: counting the positions with long segments.}
Fix $x$. A segment in \eqref{slicelongset} has Euclidean length at least $\ell$ and direction
\[
\theta(t)=\arctan(Rt),\qquad
\frac cR\leq\theta'(t)=\frac R{1+R^2t^2}\leq\frac CR
\quad(1\leq t\leq2).
\]
Its direction therefore belongs to $\Theta(S_x)$ from Lemma~\ref{slicedirections}. Changing variables from $t$ to $\theta$ gives
\[
\mes_1\{t:(t,x)\in\mathfrak D_R\}
\leq CR\mes_1\Theta(S_x)
\leq CR\eta^{1/2}\log(1/\eta).
\]
The hypotheses of that lemma follow from \eqref{fixedrhoscales}. Integrating in $x$, we obtain
\begin{equation}\label{slicelongmeasure}
\mes_2(\mathfrak D_R)
\leq CR\eta^{1/2}\log(1/\eta)\leq R^{-2\sigma}.
\end{equation}
For the last inequality, use $\eta\leq R^{-20}$ and the monotonicity of $u^{1/2}\log(1/u)$ for small $u$.

At the integer position $n$, the relevant point of this parameter set is $(n/R,x_n)$. Hence Lemma~\ref{sliceslowcount}, together with \eqref{slicelongdegree}--\eqref{slicelongmeasure}, gives
\begin{equation}\label{sliceexceptionalpositions}
\#\mathcal J\leq C(c_\rho,A)B^CR^{1-\sigma},\qquad
\mathcal J=\{R\leq n<2R:(n/R,x_n)\in\mathfrak D_R\}.
\end{equation}

\emph{Step 3: positions outside $\mathcal J$.}
Fix $n\notin\mathcal J$. On each affine piece of the wrapped graph, the intersection with $S_{x_n}$ consists of at most $CB^C$ intervals and points, by the one-dimensional component bound. Each interval has vertical length at most $2\ell$. Otherwise it contains a closed subsegment of vertical length $\ell$, which would put $(n/R,x_n)$ in $\mathfrak D_R$.

Since the slope is $n$, each such interval has $\alpha$-length at most $2\ell/n$. There are at most $n+2$ affine pieces, and isolated points have measure zero. Therefore
\begin{equation}\label{sliceshortpositions}
\mes_1(E_n)\leq CB^C\ell\qquad(n\notin\mathcal J),
\end{equation}
and summing over these positions contributes at most $CB^CR\eta^{1/4}$.

\emph{Step 4: positions in $\mathcal J$.}
For each $n\in\mathcal J$, apply Lemma~\ref{sliceoneposition} to $S_{x_n}$. Using \eqref{sliceexceptionalpositions} for the first term and $\#\mathcal J\leq R$ for the second, we get
\begin{align*}
\sum_{n\in\mathcal J}\mes_1(E_n)
&\leq CB^C\left(\frac{\#\mathcal J}{R}
       +R\eta^{1/2}\#\mathcal J\right)\\
&\leq C(c_\rho,A)B^CR^{-\sigma}+CB^CR^2\eta^{1/2}.
\end{align*}
Together with Step 3, this proves \eqref{fixedrhobound}. For negative positions, write $n=-m$ and replace $(\rho,\alpha)$ by $(-\rho,-\alpha)$. Reflection of the $\alpha$ coordinate of $S$ preserves the section areas and changes the degree by at most an absolute factor, so the same estimate applies.
\end{proof}

The interval gluing needed below follows directly from Lemma~\ref{gluing}. We apply it to the nearest-neighbor operator $H_0-EW$, whose hopping satisfies $|S(m,n)|\leq e^2e^{-2|m-n|}$ for $m\ne n$, uniformly in $E$ and $\lambda$.

\subsection{High-energy decay at a fixed potential frequency}\label{fixedhighproof}
We apply the estimates of Section~\ref{sliceestimates} to the finite candidate spectra arising in the localization argument.

\begin{theorem}[High-energy exponential decay]\label{decayhigh}
Under the assumptions of Theorem~\ref{mainhigh}, put $E_*=E_0(v,w,c_0,A)+4$, independently of $\lambda$. For every fixed $\lambda\in\R$, $\rho$ and $(x_0,y_0)$ as in that theorem, there is a set
\[
\widetilde{\mathcal R}_\rho\subset\DC^\rho_{c_0,A}(y_0),
\qquad \mes_1\widetilde{\mathcal R}_\rho=0,
\]
such that, for every $\alpha\in\DC^\rho_{c_0,A}(y_0)\setminus\widetilde{\mathcal R}_\rho$, every $|E|\geq E_*$, and every generalized eigenvector satisfying
\[
H_0(\lambda,x_0,\rho)\psi=EW(y_0,\alpha)\psi,
\qquad |\psi(n)|\leq C_\psi(1+|n|)^{C'_\psi},
\]
with finite $C_\psi,C'_\psi>0$, one has
\[
|\psi(n)|\leq\widetilde C_{\psi,\rho,\alpha}e^{-|n|/18},\qquad n\in\Z,
\]
for a finite constant $\widetilde C_{\psi,\rho,\alpha}>0$ independent of $n$. The exceptional set may depend on $\rho,\lambda,x_0,y_0,v,w,c_0,A$, but is independent of $E$, $\psi$ and its growth constants.
\end{theorem}
\begin{proof}
Fix $\lambda$, $\rho$ and $(x_0,y_0)$ as in the statement. We first obtain decay on a large annulus from avoidance of finitely many candidate energies, and then exclude the weight frequencies for which this avoidance fails.

\emph{Step 1: candidate energies from a central interval.}
Fix a layer $r\in\N$ and work with
\[
\omega=(\rho,\alpha),\qquad
\alpha\in\DC^\rho_{c_0,A}(y_0)\cap\mathcal D_r(y_0).
\]
Denote by $I_s$ the interval $[s,s+1)$ containing $E$. Since $|E|\geq E_*$,
\[
[s-1,s+2]\subset\{E':|E'|\geq E_0\}.
\]
At the scales chosen below, Proposition~\ref{ldthigh} applies throughout this enlarged interval and gives the uniform $y$-section estimate \eqref{ldtsection} for the fixed $\lambda$. Choose an arbitrary nonzero generalized eigenvector $\psi$ and normalize it by $\psi(k_0)=1$ for some $k_0\in\Z$. Then
\[
\bigl([\lambda v(x_0+n\rho)-Ew(y_0+n\alpha)]\delta_{n,n'}+\Delta\bigr)\psi=0.
\]

Let $N$ be a sufficiently large integer satisfying the conditions of Proposition~\ref{ldthigh}; for the polynomial approximation below, also take $N\geq|\lambda|+|k_0|+|s|+4$. Then $|\lambda|+|E'|+1\leq N$ for every $E'\in[s-1,s+2]$. Denote $N_2=N^{C_2}$, where $C_2$ is chosen sufficiently large that the degree condition in Lemma~\ref{counting} holds and $C_2\delta>1$.

For frequencies in this fixed layer, the Diophantine conditions give
\begin{equation}\label{finitefreq4}
\|l\cdot\omega\|_\T\geq c_0|l|^{-A},
\qquad l\in\Z^2\setminus\{0\},\quad |l|\leq N_2,
\end{equation}
and
\begin{equation}\label{finiteweight4}
\min_{1\leq i\leq m}\|y_0+(k_0+k)\alpha-y_i\|_\T
\geq r^{-1}|k_0+k|^{-A},
\quad |k|\leq N_2,\quad k\ne-k_0,
\end{equation}
at every scale. At $k=-k_0$, use $w(y_0)>0$.

The exceptional set $X_N(E):=X_N(\lambda,E,\omega)$ of Proposition~\ref{ldthigh} satisfies \eqref{ldtsection}. Moreover,
\[
\|G_{[-N,N]}\|\leq e^{N^b},\qquad
|G_{[-N,N]}(n_1,n_2)|\leq10e^{-\frac12|n_1-n_2|}
\quad\text{for }|n_1-n_2|>N/10,
\]
for all $(x,y)\notin X_N(E)$. We suppress the fixed $\lambda$ and the current frequency $\omega=(\rho,\alpha)$ in this notation. For sufficiently large $N$, we have $100N\leq N_2$. Covering the positive and negative iterates by a bounded number of blocks of length $L=\lfloor N_2/10\rfloor$ and applying Lemma~\ref{counting} on each block, using $10L\leq N_2$, we obtain
\[
\#\bigl\{|j|\leq N_2:(x_0+k_0\rho+j\rho,
y_0+k_0\alpha+j\alpha)\in X_N(E)\bigr\}
\leq C N_2^{1-\delta}.
\]
All phase coordinates are taken modulo one. Hence, for some fixed $c>0$ and all sufficiently large $N$, there exists an interval $J\subset[0,N_2]$ of size at least $cN_2^\delta>N$, such that
\[
(x_0+k_0\rho+j\rho,y_0+k_0\alpha+j\alpha)\notin X_N(E)
\quad\text{for all }j\in J\cup(-J).
\]
Suppose that the center of $J$ is $j_0$ with $|j_0|\leq N_2$, taking an integer nearest the center if necessary, and denote $[a,b]\cap\Z$ by $\Lambda$. The Poisson formula is
\begin{equation}\label{poisson}
\begin{split}
\psi(k)={}&G_\Lambda(\lambda,E,x_0,y_0,\omega)(k,a)\psi(a-1)\\
&+G_\Lambda(\lambda,E,x_0,y_0,\omega)(k,b)\psi(b+1),\qquad k\in\Lambda,
\end{split}
\end{equation}
for $E\notin\Spec H_\Lambda$, where $H_\Lambda=W_\Lambda^{-1}H_{0,\Lambda}$ is the finite weighted operator. It follows by restricting $(H_0-EW)\psi=0$ to $\Lambda$ and using the negative nearest-neighbor hopping in $\Delta$.

Take $\Lambda=[k_0+j_0-N,k_0+j_0+N]\cap\Z$. We obtain
\begin{align*}
\psi_{k_0+j_0+1}
&=\psi_{k_0+j_0-N-1}G_\Lambda(k_0+j_0-N,k_0+j_0+1)\\
&\quad+\psi_{k_0+j_0+N+1}G_\Lambda(k_0+j_0+N,k_0+j_0+1).
\end{align*}
Therefore, with the invariance of Green's function under translation, we have
\begin{align*}
|\psi_{k_0+j_0+1}|
&\leq\sum_{\sigma\in\{-1,1\}}|\psi_{k_0+j_0+\sigma(N+1)}|\\
&\qquad\cdot|G_N(\lambda,E,x_0+(k_0+j_0)\rho,
y_0+(k_0+j_0)\alpha,\omega)(1,\sigma N)|\\
&\leq10C_\psi(1+|k_0|+N_2+N+1)^{C'_\psi}
\bigl(e^{-(N+1)/2}+e^{-(N-1)/2}\bigr)\\
&\leq20C_\psi(1+|k_0|+N_2+N+1)^{C'_\psi}e^{-\frac12(N-1)}\\
&\leq\frac12(|k_0|+1)^{C'_\psi}e^{-N/3}.
\end{align*}
Similarly, take $\Lambda=[k_0-j_0-N,k_0-j_0+N]\cap\Z$. We have
\[
|\psi_{k_0-j_0-1}|\leq\frac12(|k_0|+1)^{C'_\psi}e^{-N/3}.
\]
Both estimates hold for sufficiently large $N$.

Now substitute $v$ and $w$ by the polynomials $v_1$ and $w_1$ which are the same as before, and denote, on a finite interval $\Lambda\subset[k_0-N_2,k_0+N_2]$,
\begin{align*}
\widetilde H_{0,\Lambda}(\lambda,x_0,\rho)
&=\bigl(\lambda v_1(\{x_0+n\rho\})\delta_{n,n'}+\Delta\bigr)_\Lambda,\\
\widetilde W_\Lambda(y_0,\alpha)
&=\bigl(w_1(\{y_0+n\alpha\})\delta_{n,n'}\bigr)_\Lambda,\\
\widetilde G_\Lambda(\lambda,E,x_0,y_0,\omega)
&=(\widetilde H_{0,\Lambda}-E\widetilde W_\Lambda)^{-1},
\end{align*}
the last expression being used when the inverse exists. Braces denote fractional parts; all polynomial evaluations below use these representatives.

Moreover, similar to \eqref{poisson}, $\psi$ can be represented, whenever $\widetilde G_\Lambda$ exists, as
\begin{align}
R_\Lambda\psi
&=\widetilde G_\Lambda(\widetilde H_{0,\Lambda}-E\widetilde W_\Lambda)R_\Lambda\psi\notag\\
&=-\widetilde G_\Lambda R_\Lambda(H_0-EW)R_{\Lambda^c}^*R_{\Lambda^c}\psi\notag\\
&\quad+\widetilde G_\Lambda\bigl[(\widetilde H_{0,\Lambda}-H_{0,\Lambda})
+E(W_\Lambda-\widetilde W_\Lambda)\bigr]R_\Lambda\psi\notag\\
&=\psi_{a-1}\widetilde G_\Lambda\delta_a+\psi_{b+1}\widetilde G_\Lambda\delta_b\notag\\
&\quad+\sum_{a\leq l\leq b}\bigl[\lambda(v_1-v)(x_0+l\rho)
+E(w-w_1)(y_0+l\alpha)\bigr]\psi_l\widetilde G_\Lambda\delta_l.\label{approxpoisson}
\end{align}

Therefore, taking $\Lambda=[k_0-j_0,k_0+j_0]\cap\Z$, if its approximate operator is invertible, we obtain
\begin{align*}
1=|\psi(k_0)|
&\leq\|\widetilde G_\Lambda\|\biggl(\sum_{l=k_0-j_0}^{k_0+j_0}
\bigl|\lambda(v_1-v)(x_0+l\rho)+E(w-w_1)(y_0+l\alpha)\bigr|\,|\psi_l|\\
&\hspace{10em}+|\psi_{k_0-j_0-1}|+|\psi_{k_0+j_0+1}|\biggr)\\
&\leq\|\widetilde G_\Lambda\|\bigl[(|k_0|+1)^{C'_\psi}e^{-N/3}\\
&\hspace{4em}+(2N_2+1)(|\lambda|+|s|+1)e^{-N^3}
C_\psi(1+|k_0|+N_2)^{C'_\psi}\bigr]\\
&<\|\widetilde G_\Lambda\|(|k_0|+1)^{C'_\psi}e^{-N/4}
\end{align*}
for sufficiently large $N$.

Combining the second Diophantine condition \eqref{finiteweight4} with the properties of the weight $w$, we have
\[
w(y_0+(k_0+k)\alpha)
\geq c_w\min_{1\leq i\leq m}\|y_0+(k_0+k)\alpha-y_i\|_\T^{C_w}
\]
when $k\ne-k_0$, where $C_w$ is the maximal order of a zero of $w$ and $c_w>0$. Together with $w(y_0)>0$, this yields, for all $|k|\leq N_2$,
\[
w(y_0+(k_0+k)\alpha)
\gtrsim_{w,y_0,r}(1+|k_0|+N_2)^{-AC_w}.
\]
Since $|w-w_1|<e^{-N^3}$, we have $w_1\geq w/2>0$ throughout this finite range for sufficiently large $N$. Therefore
\[
\|\widetilde W_\Lambda^{-1/2}\|^2
\leq2\|W_\Lambda^{-1/2}\|^2
\lesssim_{w,y_0,r}(1+|k_0|+N_2)^{AC_w}.
\]
This implies that
\[
\widetilde H_{1,\Lambda}(\lambda,x_0,y_0,\omega)
=\widetilde W_\Lambda^{-1/2}\widetilde H_{0,\Lambda}\widetilde W_\Lambda^{-1/2}
\]
is a finite, self-adjoint matrix. Since
\[
\widetilde G_\Lambda
=\widetilde W_\Lambda^{-1/2}(\widetilde H_{1,\Lambda}-E)^{-1}
\widetilde W_\Lambda^{-1/2},
\]
we have
\begin{align*}
\dist(E,\Spec\widetilde H_{1,\Lambda})
&=\|(\widetilde H_{1,\Lambda}-E)^{-1}\|^{-1}\\
&\leq\|\widetilde W_\Lambda^{-1/2}\|^2\|\widetilde G_\Lambda\|^{-1}\\
&<(|k_0|+1)^{AC_w+C'_\psi}e^{-N/5}.
\end{align*}
If the approximate operator is singular, $E\in\Spec\widetilde H_{1,\Lambda}$ and the same distance estimate is immediate.

Let $N_3=N^{C_3}$, where the fixed integer $C_3>2$ is sufficiently large depending on the constants in the semi-algebraic estimates and on $C_2$. Denote
\[
\mathcal E_{\omega,s}
=\bigcup_{0\leq j\leq N_2}
\bigl(\Spec\widetilde H_{1,[k_0-j,k_0+j]}(\lambda,x_0,y_0,\omega)\cap[s-1,s+2]\bigr).
\]
The scale $N$ and the fixed center $k_0$ are suppressed in this notation. Since $E\in[s,s+1)$, the preceding estimate gives $E'\in\mathcal E_{\omega,s}$ with
\begin{equation}\label{candidateapprox}
|E-E'|<(|k_0|+1)^{AC_w+C'_\psi}e^{-N/5}.
\end{equation}
This holds for sufficiently large $N$.

\emph{Step 2: decay on an annulus under the avoidance condition.}
Let us first assume that
\begin{equation}\label{avoidclaim}
\begin{split}
&(x_0+k_0\rho+n\rho,y_0+k_0\alpha+n\alpha)
\notin\bigcup_{E'\in\mathcal E_{\omega,s}}X_N(E')\pmod1,\\
&\hspace{13em}\text{for all }N_3^{1/2}\leq|n|\leq2N_3.
\end{split}
\end{equation}
holds.

Then by the Neumann series, on $\Lambda=[k_0+n-N,k_0+n+N]\cap\Z$,
\begin{align*}
G_\Lambda(E)
&=\bigl[H_{0,\Lambda}-E'W_\Lambda+(E'-E)W_\Lambda\bigr]^{-1}\\
&=G_\Lambda(E')\sum_{l=0}^{\infty}
\bigl[(E-E')W_\Lambda G_\Lambda(E')\bigr]^l.
\end{align*}
The smallness condition of Lemma~\ref{perturb} is satisfied, since
\[
C(2N+1)(|k_0|+1)^{AC_w+C'_\psi}
e^{-N/5+2N^b+N/10}\longrightarrow0.
\]
We obtain the off-diagonal decay for Green's function with $E'$ replaced by $E$, together with the norm bound,
\[
\|G_{[k_0+n-N,k_0+n+N]}(E)\|\leq2e^{N^b},
\]
\[
|G_{[k_0+n-N,k_0+n+N]}(E)(n_1,n_2)|
\leq11e^{-\frac12|n_1-n_2|},\qquad |n_1-n_2|>N/10.
\]

Let
\[
\widetilde\Lambda=\bigcup_{N_3^{1/2}\leq|n|\leq2N_3}
[k_0+n-N,k_0+n+N]\cap\Z.
\]
Put $r_3=\lceil N_3^{1/2}\rceil$. Each of the two connected components $I=[a_I,b_I]\cap\Z$ has size
\[
|I|=2N_3-r_3+2N+1\asymp N_3.
\]
Their restrictions form a direct sum because the hopping is nearest-neighbor. For $k\in I$, choose
\[
c(k)=\min\{\max\{k,a_I+N\},b_I-N\},\qquad
I_k=[c(k)-N,c(k)+N]\cap\Z.
\]
The centers remain in the range of \eqref{avoidclaim}, and
$I_k\subset I$, $|I_k|=2N+1$ and $\dist(k,I\setminus I_k)\geq N+1$.
Lemma~\ref{gluing} therefore applies with $M=2N+1$ and $\gamma=1/2$, giving $\|G_I(E)\|\leq CMe^{M^b}<e^M$. Its loss in the decay rate is at most
\[
C\frac{\log(2N+1)}{2N+1}
+C\frac{2N+1}{|I|}
=o(1),
\]
since $N_3=N^{C_3}$ and $C_3>2$. It is less than $1/4$ for all sufficiently large $N$, and hence
\[
|G_I(E)(n_1,n_2)|\leq e^{-|n_1-n_2|/4},
\qquad |n_1-n_2|>|I|/10.
\]

For $N_3/2\leq|j|\leq N_3$, the distance from $k_0+j$ to either endpoint of its component is at least $|j|-r_3+N>|I|/10$ for large $N$. Therefore the preceding decay estimate applies to both terms in \eqref{poisson}, giving
\begin{align*}
|\psi(k_0+j)|
&\leq2C_\psi(1+|k_0|+3N_3)^{C'_\psi}
e^{-\frac14(|j|-N_3^{1/2}-2)}\\
&\leq e^{-\frac1{17}|j|}.
\end{align*}
\emph{Step 3: semi-algebraic sets of candidate resonances.}
We now estimate the weight frequencies for which \eqref{avoidclaim} fails. Keep $r,k_0,s$ and $N$ fixed. For $0\leq j\leq N_2$, let
\[
\begin{split}
G_{j,\rho}=\bigl\{(\alpha,E',x,y):\ &\alpha\in\T,\ E'\in[s-1,s+2],\ (x,y)\in\T^2,\\
&(\rho,\alpha)\text{ satisfies \eqref{finitefreq4} and \eqref{finiteweight4}},\\
&E'\in\Spec\widetilde H_{1,[k_0-j,k_0+j]}(\lambda,x_0,y_0,\rho,\alpha),\\
&(x,y)\in X_N(\lambda,E',\rho,\alpha)\bigr\},
\end{split}
\]
and set $S_{j,\rho}=\Proj_{(\alpha,x,y)}G_{j,\rho}$. This projection collects the bad remote phases associated with the spectrum of the central interval $[k_0-j,k_0+j]$.

For each $\alpha$ satisfying the finite frequency conditions, the weight lower bound in Step 1 makes the central approximate weight matrix positive once $N$ is sufficiently large. Its self-adjoint central matrix has at most $2j+1$ distinct eigenvalues. These eigenvalues depend on $\alpha$ and the fixed central phases $(x_0,y_0)$; the variables $x,y$ enter only through $X_N$. Thus, for each fixed $(\alpha,x)$, the $y$-section of $S_{j,\rho}$ is a union of at most $2j+1$ sections from \eqref{ldtsection}. Consequently
\begin{equation}\label{slicephasebound}
\sup_{\alpha,x}\mes_1\{y:(\alpha,x,y)\in S_{j,\rho}\}
\leq(2N_2+1)e^{-N^\kappa}<e^{-N^\kappa/2}=: \eta_N.
\end{equation}
For other values of $\alpha$ the section is empty. Integrating in $\alpha$ gives the hypothesis needed in Lemma~\ref{fixedrhoavoid}:
\begin{equation}\label{sliceareabound}
\sup_{x\in\T}\mes_2\{(\alpha,y):(\alpha,x,y)\in S_{j,\rho}\}\leq\eta_N.
\end{equation}

We next bound the degree of $S_{j,\rho}$. Resolving the distances modulo one in \eqref{finitefreq4}--\eqref{finiteweight4} gives $O(N_2^3+mN_2^2)$ affine inequalities, since $|k_0|\leq N$. On this set the central weights are positive, so the spectral condition is equivalent to
\[
\det(\widetilde H_{0,[k_0-j,k_0+j]}-E'\widetilde W_{[k_0-j,k_0+j]})=0,
\qquad (E'-(s-1))(E'-(s+2))\leq0.
\]
Partition the $\alpha$ interval at the integer crossings of $y_0+l\alpha$ for $|l-k_0|\leq N_2$. There are polynomially many pieces. On each piece the determinant is a polynomial in $(E',\alpha)$ of degree at most $C(2j+1)N^4$. Since $\rho$ is fixed, the values $v_1(\{x_0+l\rho\})$ are coefficients of this polynomial.

The remaining condition, $(x,y)\in X_N(\lambda,E',\rho,\alpha)$, is jointly semi-algebraic in $(E',\alpha,x,y)$ by Proposition~\ref{ldthigh}. Choose $C_2\geq\max\{4,C_1\}$, with $C_1$ large enough for this joint description. Applying Lemma~\ref{projection} to eliminate $E'$ gives
\[
\deg S_{j,\rho}\leq B_N=N^D
\]
for a fixed exponent $D$ independent of $N,j$. We have now verified both the section-measure and degree hypotheses of Lemma~\ref{fixedrhoavoid}.

\emph{Step 4: exclusion of weight frequencies.}
For the positions in \eqref{avoidclaim}, put $q=k_0+n$. Since $|k_0|\leq N$ and $C_3>2$,
\[
\tfrac12N_3^{1/2}\leq|q|\leq3N_3
\]
for large $N$. Cover this range by $O(\log N)$ positive and negative dyadic blocks $R\leq|q|<2R$. On each block apply Lemma~\ref{fixedrhoavoid} to $S_{j,\rho}$ with $B=B_N$ and $\eta=\eta_N$. Its arithmetic hypothesis is \eqref{rhodc}, so we take $c_\rho=c_0$. The scale conditions hold once $C_3/2>DC_*$ and $N$ is large: the smallest $R$ then exceeds $B_N^{C_*}$, and $\eta_N\leq R^{-20}$ even on the largest block.

With $\sigma=1/(A+2)$, summing the conclusion of that lemma over these blocks gives
\begin{align*}
&\mes_1\bigl\{\alpha:(\alpha,x_0+(k_0+n)\rho,y_0+(k_0+n)\alpha)\in S_{j,\rho}\\
&\hspace{9em}\text{for some }N_3^{1/2}\leq|n|\leq2N_3\bigr\}\\
&\quad\leq C(c_0,A)N^{CD}\log N\left(
N_3^{-\sigma/2}+N_3e^{-N^\kappa/8}+N_3^2e^{-N^\kappa/4}\right).
\end{align*}
The first term uses the smallest block size; the other two use the largest. Define
\[
\begin{split}
\mathcal R^{(r)}_{\rho,k_0,s,N}
=\bigcup_{0\leq j\leq N_2}\bigl\{\alpha:\ &
(\alpha,x_0+(k_0+n)\rho,y_0+(k_0+n)\alpha)\in S_{j,\rho}\\
&\text{for some }N_3^{1/2}\leq|n|\leq2N_3\bigr\}.
\end{split}
\]
Within the finite frequency conditions, this is precisely the failure of \eqref{avoidclaim}. There are at most $N_2+1$ values of $j$. Choose $C_3$ sufficiently large that, in addition to $C_3>2$ and $C_3/2>DC_*$,
\[
\frac{C_3\sigma}{2}>C_2+CD+12.
\]
The first term in the preceding measure bound is then summable after the union over $j$, and the exponential terms are smaller than every fixed negative power of $N$. Hence
\begin{equation}\label{slicesummable}
\mes_1\mathcal R^{(r)}_{\rho,k_0,s,N}<N^{-10}
\end{equation}
for all sufficiently large $N$. These sets depend on the central candidate spectra and on $r,k_0,s$, and are independent of the chosen generalized eigenvector.

\emph{Step 5: the null set and exponential decay.}
Let $s$ range over the integers for which $[s-1,s+2]\subset\{|E'|\geq E_0\}$, and put
\[
\widetilde{\mathcal R}_\rho
=\DC^\rho_{c_0,A}(y_0)\cap
\bigcup_{\substack{r\in\N,\ k_0\in\Z\\s\ \text{as above}}}
\limsup_{N\to\infty}\mathcal R^{(r)}_{\rho,k_0,s,N}.
\]
By \eqref{slicesummable} and Borel--Cantelli, each limsup has measure zero. The union is countable, so $\mes_1\widetilde{\mathcal R}_\rho=0$. Every $|E|\geq E_*$ lies in one of the intervals $[s,s+1)$ used here. The set $\widetilde{\mathcal R}_\rho$ is therefore chosen simultaneously for all the energies and generalized eigenvectors in the theorem.

For the fixed $\lambda,\rho,x_0,y_0$, choose any $\alpha\in\DC^\rho_{c_0,A}(y_0)\setminus\widetilde{\mathcal R}_\rho$ and a layer $r$ with $\alpha\in\mathcal D_r(y_0)$. Arbitrarily choose a nonzero generalized eigenvector $\psi$ with $E\in[s,s+1)$ and $\psi(k_0)\ne0$, and normalize it as above. By the choice of $\alpha$ and $r$, conditions \eqref{finitefreq4} and \eqref{finiteweight4} hold at every scale. For sufficiently large $N$ we have $\alpha\notin\mathcal R^{(r)}_{\rho,k_0,s,N}$, and \eqref{avoidclaim} holds. The preceding argument gives
\[
|\psi(k_0+j)|\leq e^{-|j|/17},
\qquad N^{C_3}/2\leq|j|\leq N^{C_3}.
\]
For large $N$, the annuli at consecutive integer scales overlap, since $(N+1)^{C_3}/2\leq N^{C_3}$, and hence cover both tails. Absorbing the shift $k_0$, the finitely many remaining sites, and the normalization of $\psi$ into a finite constant $\widetilde C_{\psi,\rho,\alpha}>0$, we obtain
\[
|\psi(j)|\leq\widetilde C_{\psi,\rho,\alpha}e^{-|j|/18},\qquad j\in\Z.
\]

\end{proof}

For a fixed $\rho$ satisfying $\|k\rho\|_\T\geq c_\rho|k|^{-A}$, the mixed Diophantine slices exhaust almost every weight frequency as their constants decrease. Indeed, for $q\ne0$,
\[
\mes_1\{\alpha:\|p\rho+q\alpha\|_\T<c|(p,q)|^{-A}\}
\leq2c|(p,q)|^{-A}.
\]
The sum over $(p,q)$ with $q\ne0$ is at most $C_Ac$ because $A>2$; the terms $q=0$ satisfy the required inequality whenever $c\leq c_\rho$. Lemma~\ref{orbit} removes only a null set for the weight condition. Consequently
\[
\mes_1\left(\T\setminus\bigcup_{m\geq1}
\DC^\rho_{c_\rho2^{-m},A}(y_0)\right)=0.
\]
Theorem~\ref{decayhigh} can thus be applied for almost every $\alpha\in\T$, with the energy threshold corresponding to a mixed Diophantine class containing $(\rho,\alpha)$.

\begin{proof}[Proof of Theorem~\ref{mainhigh}]
Theorem~\ref{decayhigh} gives exponential decay of all polynomially bounded generalized eigenvectors in the high-energy region for every allowed weight frequency outside the null set in the fixed-$\rho$ slice. From the relation between $H$ and $H_1$ discussed in Section~2.2, the same holds for $H_1$. Indeed, if $u'$ is a polynomially bounded generalized eigenvector of $H_1$, then $u=W^{-1/2}u'$ is polynomially bounded by Lemma~\ref{orbit}; exponential decay of $u$ gives that of $u'=W^{1/2}u$ because $w$ is bounded.

By Theorem~\ref{schnol}(a), the canonical spectral measure of $H_1$ is supported on polynomially bounded generalized eigenvalues. The preceding exponential-decay result turns these generalized eigenvectors in the stated region into genuine $\ell^2$ eigenvectors: their coordinatewise image is $Eu'\in\ell^2$, so the maximal-domain description \eqref{domain} applies. Hence the spectral measure in this region is supported on eigenvalues with exponentially decaying eigenvectors. Since $H_1$ is self-adjoint on a separable Hilbert space, these eigenvalues are countable; thus the spectral type is pure point. Every genuine $\ell^2$ eigenvector of $H_1$ is bounded, so the same transformation and decay argument applies to it. From the relation between $H$ and $H_1$ again, $H$ has pure point spectrum with exponentially localized eigenvectors in the corresponding region. \end{proof}

\subsection{Large-coupling decay with varying frequencies}\label{varyingproof}
We now allow both frequencies to vary. Proposition~\ref{ldtlarge} gives the phase-measure estimate at every energy, and Lemma~\ref{transversality} gives the frequency exclusion.

\begin{proposition}[Large-coupling exponential decay]\label{decaylarge}
Under the assumptions of Theorem~\ref{mainlarge}, fix $|\lambda|>\lambda_0$ and $(x_0,y_0)\in\T^2$ with $w(y_0)>0$. There is a null set $\widetilde{\mathcal R}\subset\DC_{c_0,A}(y_0)$ such that, for every $\omega=(\rho,\alpha)$ outside this set and every $E\in\R$, each solution of
\[
H_0(\lambda,x_0,\rho)\psi=EW(y_0,\alpha)\psi,
\qquad |\psi(n)|\leq C_\psi(1+|n|)^{C'_\psi},
\]
with finite $C_\psi,C'_\psi>0$, satisfies
\[
|\psi(n)|\leq\widetilde C_{\psi,\omega}e^{-|n|/18},\qquad n\in\Z,
\]
for some finite $\widetilde C_{\psi,\omega}>0$. The exceptional set is independent of $E,\psi$ and its growth constants.
\end{proposition}
\begin{proof}
\emph{Step 1: candidate energies and decay on an annulus.}
Fix a weight layer $r\in\N$, an energy interval $I_s=[s,s+1)$ and a center $k_0\in\Z$. Work with
\[
\omega\in\DC_{c_0,A}\cap(\T\times\mathcal D_r(y_0)),
\]
and normalize a nonzero solution with $E\in I_s$ by $\psi(k_0)=1$. Use $N_2=N^{C_2}$ and $N_3=N^{C_3}$ as in Theorem~\ref{decayhigh}, taking $N\geq|\lambda|+|k_0|+|s|+4$ and increasing $C_3$ below. The finite conditions \eqref{finitefreq4}--\eqref{finiteweight4} hold on this layer.

Write $Y_N(E')=Y_N(\lambda,E',\omega)$. Proposition~\ref{ldtlarge} gives $\mes_2Y_N(E')<e^{-N^\kappa}$ and \eqref{ldtnorm}--\eqref{ldtdecay} throughout $[s-1,s+2]$, with no restriction on $s$. These are the estimates used in the central-interval argument of Step 1 of Theorem~\ref{decayhigh}. Use the absolute polynomial approximations from Proposition~\ref{ldtlarge}, of degree less than $N^4$ and error less than $e^{-N^3}$. The weight lower bound on the fixed layer ensures that $\widetilde W_\Lambda$ is positive on the central intervals, so the corresponding matrices $\widetilde H_{1,\Lambda}$ are self-adjoint. With $Y_N$ in place of $X_N$, the same argument gives a candidate $E'\in\mathcal E_{\omega,s}$ satisfying \eqref{candidateapprox}, where $\mathcal E_{\omega,s}$ is the union of central spectra defined there.

Suppose now that
\begin{equation}\label{largeavoid}
\begin{split}
&(x_0+(k_0+n)\rho,y_0+(k_0+n)\alpha)
\notin\bigcup_{E'\in\mathcal E_{\omega,s}}Y_N(E'),\\
&\hspace{12em}N_3^{1/2}\leq|n|\leq2N_3.
\end{split}
\end{equation}
The perturbation from $E'$ to $E$ has norm at most $\|w\|_\T|E-E'|$. By \eqref{candidateapprox}, it satisfies the smallness condition checked in Step 2 of Theorem~\ref{decayhigh}. The same transfer of local inverse bounds, gluing and Poisson estimate therefore give
\begin{equation}\label{largeannulus}
|\psi(k_0+j)|\leq e^{-|j|/17},
\qquad N_3/2\leq|j|\leq N_3.
\end{equation}
It remains to establish \eqref{largeavoid} outside a summable family of frequency sets.

\emph{Step 2: semi-algebraic sets of candidate resonances.}
For $0\leq j\leq N_2$, define
\[
\begin{split}
G_j=\bigl\{(\omega,E',x,y):\ &\omega\in\T^2,\ E'\in[s-1,s+2],\ (x,y)\in\T^2,\\
&\omega\text{ satisfies \eqref{finitefreq4} and \eqref{finiteweight4}},\\
&E'\in\Spec\widetilde H_{1,[k_0-j,k_0+j]}(\lambda,x_0,y_0,\omega),\\
&(x,y)\in Y_N(\lambda,E',\omega)\bigr\},
\end{split}
\]
and let $S_j=\Proj_{(\omega,x,y)}G_j$. For fixed $\omega$, the central matrix has at most $2j+1$ distinct eigenvalues, independently of the remote phases $x,y$. Consequently
\[
\mes_2\{(x,y):(\omega,x,y)\in S_j\}
\leq(2N_2+1)e^{-N^\kappa}<e^{-N^\kappa/2}=: \eta_N.
\]
The section is empty if the finite frequency conditions fail. Fubini's theorem gives
\begin{equation}\label{largefourmeasure}
\mes_4 S_j\leq\eta_N.
\end{equation}

The degree calculation is the one used for $S_{j,\rho}$, now with both frequency coordinates variable. Resolving the distances modulo one gives $O(N_2^3+mN_2^2)$ affine conditions. Since the central approximate weights are positive, the spectral condition becomes
\[
\det(\widetilde H_{0,[k_0-j,k_0+j]}-E'\widetilde W_{[k_0-j,k_0+j]})=0.
\]
Partition each frequency coordinate at the integer crossings of $x_0+l\rho$ and $y_0+l\alpha$, respectively, for $|l-k_0|\leq N_2$. There are $O(N_2^4)$ rectangles in the resulting product partition. On each rectangle the determinant has degree at most $C(2j+1)N^4$ in $(E',\omega)$. The remaining condition involving $Y_N$ is jointly semi-algebraic by Proposition~\ref{ldtlarge}. Lemma~\ref{projection}, applied in the five variables $(\rho,\alpha,E',x,y)$, therefore gives
\begin{equation}\label{largedegree}
\deg S_j\leq B_N=N^{D_{\rm c}}
\end{equation}
for a fixed exponent $D_{\rm c}$, independent of $j,N,\lambda,E$.

\emph{Step 3: exclusion of frequency pairs.}
Apply Lemma~\ref{transversality} with $n=2$, $B=B_N$, $\eta=\eta_N$ and $\eps=N_3^{-1/10}$. Its scale hypotheses hold for large $N$. The decomposition $S_j=S_{j,1}\cup S_{j,2}$ satisfies
\begin{equation}\label{largeprojection}
\mes_2(\Proj_\omega S_{j,1})\leq B_N^C N_3^{-1/10}.
\end{equation}
This is a single decomposition for the fixed set $S_j$, so its projection contains the exceptional frequencies coming from $S_{j,1}$ at every position in \eqref{largeavoid}.

For $N_3^{1/2}\leq|n|\leq2N_3$, put $q=k_0+n$. Partition the frequency torus into squares of side $(2N_3)^{-1}$, and subdivide where either $x_0+q\rho$ or $y_0+q\alpha$ crosses an integer. Since $|q|\leq3N_3$, there are $O(N_3^2)$ pieces for each $n$. On each piece the map
\[
\omega\longmapsto
(\omega,x_0+q\rho-p_1,y_0+q\alpha-p_2)
\]
parametrizes part of an affine two-plane $L$, with constant integers $p_1,p_2$. For either frequency-coordinate unit vector $e_i$,
\[
\|\Proj_L e_i\|=(1+q^2)^{-1/2}
\leq3N_3^{-1/2}<\eps/100
\]
when $N$ is large. Lemma~\ref{transversality} gives
\[
\mes_2(S_{j,2}\cap L)\leq B_N^C N_3^{1/10}e^{-N^\kappa/8}.
\]
Projection onto the frequency coordinates cannot increase this two-dimensional measure. Summing over $O(N_3)$ positions and $O(N_3^2)$ pieces at each position yields
\begin{align*}
&\mes_2\bigl\{\omega:(\omega,x_0+(k_0+n)\rho,y_0+(k_0+n)\alpha)\in S_{j,2}\\
&\hspace{10em}\text{for some }N_3^{1/2}\leq|n|\leq2N_3\bigr\}\\
&\qquad\leq C B_N^C N_3^{31/10}e^{-N^\kappa/8}.
\end{align*}
Combining this with \eqref{largeprojection}, and taking the union over $0\leq j\leq N_2$, defines the failure set $\mathcal R^{(r)}_{k_0,s,N}$ of \eqref{largeavoid}, within the finite frequency conditions. It satisfies
\[
\mes_2\mathcal R^{(r)}_{k_0,s,N}
\leq C(N_2+1)B_N^C
\left(N_3^{-1/10}+N_3^{31/10}e^{-N^\kappa/8}\right).
\]
Choose $C_3>10(C_2+CD_{\rm c}+12)$, increasing it if necessary to satisfy the earlier scale conditions. Then
\begin{equation}\label{largesummable}
\mes_2\mathcal R^{(r)}_{k_0,s,N}<N^{-10}
\end{equation}
for all sufficiently large $N$.

\emph{Step 4: the null set and exponential decay.}
Put
\[
\widetilde{\mathcal R}
=\DC_{c_0,A}(y_0)\cap
\bigcup_{\substack{r\in\N\\ k_0,s\in\Z}}
\limsup_{N\to\infty}\mathcal R^{(r)}_{k_0,s,N}.
\]
Below the starting threshold for fixed $r,k_0,s$, the sets may be taken to be empty. By \eqref{largesummable} and Borel--Cantelli, $\mes_2\widetilde{\mathcal R}=0$. The construction uses all the central candidate eigenvalues and all integer energy windows, so this set is independent of the energy and the generalized eigenvector.

Fix $\omega\in\DC_{c_0,A}(y_0)\setminus\widetilde{\mathcal R}$ and a layer containing $\alpha$. For any polynomially bounded solution, choose $s,k_0$ as in Step 1. Then \eqref{largeavoid}, and hence \eqref{largeannulus}, hold for all sufficiently large $N$. The annuli at consecutive integer scales overlap because $(N+1)^{C_3}/2\leq N^{C_3}$ for large $N$. Absorbing the center, the normalization and the finitely many remaining sites into a constant gives the asserted decay on $\Z$.
\end{proof}

\begin{proof}[Proof of Theorem~\ref{mainlarge}]
Take $\lambda_0$ from Proposition~\ref{ldtlarge} and the null set from Proposition~\ref{decaylarge}. For every remaining frequency pair, that proposition gives exponential decay of all polynomially bounded generalized eigenvectors at every $E\in\R$. Apply the transformation $u'=W^{1/2}u$, the orbit lower bound and the Schnol-type theorem exactly as in the proof of Theorem~\ref{mainhigh}. The canonical spectral measure is supported on exponentially decaying eigenvectors on the whole real line, and every genuine eigenvector has the same decay. This proves the full-spectrum assertions for $H$ and $H_1$.
\end{proof}

\section{Application to a Frenkel--Kontorova chain}\label{fkapplication}
For a prescribed analytic FK equilibrium, the masses can vary without changing the static configuration. We derive the resulting weighted phonon equation, compare its parameters with those of the unit-mass model and apply Theorem~\ref{mainhigh} to obtain localization. We conclude with the role of the localized modes in further studies of nonlinear vibrations.

\subsection{Analytic equilibria and the weighted phonon equation}
The Frenkel--Kontorova model describes a chain of interacting particles in a periodic substrate potential; see \cite{BK91} or Chapters~2 and~3 in \cite{BK04}. With unit masses, the Hamiltonian in a standard normalization may be written as
\begin{equation}\label{fkham}
\mathscr H_K(u,\dot u)=\sum_n\left\{\frac12\dot u_n^2+
\frac12(u_{n+1}-u_n-\rho)^2+\frac{K}{(2\pi)^2}\cos(2\pi u_n)\right\}.
\end{equation}
Here $K$ is the strength of the substrate potential, and $\rho$ is the prescribed mean interparticle spacing of the chain. The corresponding nonlinear lattice equation is
\begin{equation}\label{fkeq}
\ddot u_n=u_{n+1}+u_{n-1}-2u_n+\frac K{2\pi}\sin(2\pi u_n).
\end{equation}
An equilibrium configuration is a time-independent solution $\{\bar u_n\}$ of this equation; a ground state is an equilibrium with an additional minimizing property. Equivalently, an equilibrium satisfies
\begin{equation}\label{fkequil}
\bar u_{n+1}+\bar u_{n-1}-2\bar u_n+\frac K{2\pi}\sin(2\pi\bar u_n)=0.
\end{equation}
After introducing the discrete slope variable
\[
p_n=\bar u_n-\bar u_{n-1}-\rho,
\]
this equilibrium equation takes the standard-map form
\[
p_{n+1}=p_n-\frac K{2\pi}\sin(2\pi\bar u_n),\qquad
\bar u_{n+1}=\bar u_n+\rho+p_{n+1}.
\]
This is the standard-map description used in Aubry's theory \cite{AD83}. With the substrate period normalized to one, the chain is called \emph{incommensurate} when its mean spacing $\rho$ is irrational. Aubry's hull representation \cite[Theorem~2, parts~2(d)--(e)]{A83} parametrizes the corresponding ground states by increasing hull functions. We fix $(K,\rho)$ in the analytic regime and write the ground state as
\[
\bar u_n=h_{K,\rho}(\theta+n\rho)
=\theta+n\rho+g_{K,\rho}(\theta+n\rho),
\]
where $h_{K,\rho}:\R\to\R$ is strictly increasing and real analytic, and
\[
h_{K,\rho}(t)=t+g_{K,\rho}(t),\qquad g_{K,\rho}(t+1)=g_{K,\rho}(t).
\]
Thus $g_{K,\rho}$ is the periodic displacement from the uniformly spaced chain. Equivalently, $h_{K,\rho}(t+1)=h_{K,\rho}(t)+1$, so $h_{K,\rho}$ is a degree-one lift of a circle map. This is the analytic hull description in \cite[Section~2, equation~(7)]{PA83}. In particular, $\cos(2\pi h_{K,\rho})$ is $1$-periodic and real analytic. We write $h_K=h_{K,\rho}$ when the mean spacing is fixed.

We now introduce quasi-periodic masses. Impurity versions of the FK model, including chains with non-identical particles, have been studied in the physics literature; see again \cite{BK91,BK04}. In the present setting the mass at site $n$ is $m_n=w(y_0+n\alpha)>0$, with $w$ satisfying the assumptions above. Thus the kinetic energy in \eqref{fkham} is replaced by
\[
\sum_n\frac12m_n\dot u_n^2,
\]
while the spring interaction and the substrate potential are left unchanged. Since the masses occur only in the kinetic energy, the equilibrium equation is still determined by the potential energy. In particular, for a given substrate strength $K$, the equilibrium $\{\bar u_n\}$ satisfies the same equilibrium equation as in the mass-homogeneous FK chain. The minimizing configurations and their hull functions are therefore unchanged as well. The masses enter the problem only after one studies the dynamics near this equilibrium.

To study vibrations around the equilibrium, write
\[
u_n(t)=\bar u_n+q_n(t).
\]
Substituting this into the equation of motion with masses $m_n$ and using \eqref{fkequil}, one obtains
\begin{equation}\label{weightedmotion}
m_n\ddot q_n=q_{n+1}+q_{n-1}-2q_n+K\cos(2\pi\bar u_n)q_n+R_n(q_n),
\end{equation}
where $R(q)=(R_n(q_n))_{n\in\Z}$ and
\[
R_n(\xi)=\frac K{2\pi}\bigl[\sin(2\pi(\bar u_n+\xi))-\sin(2\pi\bar u_n)
-2\pi\cos(2\pi\bar u_n)\xi\bigr]=O(\xi^2)
\]
uniformly in $n$ near $\xi=0$.

Equivalently, if $W$ denotes the diagonal mass operator and $H_K$ denotes the linear phonon operator coming from the spring and substrate forces, then \eqref{weightedmotion} has the formal structure
\begin{equation}\label{wmotion}
W\ddot q=H_Kq+R(q),\qquad
(H_Kq)_n=q_{n+1}+q_{n-1}+[K\cos(2\pi\bar u_n)-2]q_n.
\end{equation}
Neglecting the nonlinear term $R(q)$ gives the linearized equation
\[
W\ddot q=H_Kq.
\]
Looking for time-harmonic normal modes of the linearized equation,
\[
q_n(t)=\psi_ne^{i\Omega t},
\]
where $\Omega\geq0$ denotes the vibration angular frequency, one obtains the generalized eigenvalue problem
\[
H_K\psi=-\Omega^2W\psi.
\]

It would not be enough to know localization for the unweighted operator $H_K$. Although $W$ is diagonal in the physical lattice basis, it need not be diagonal in a spectral basis of $H_K$. If $\{\phi_j\}$ is an eigenbasis of $H_K$, then the equation $H_K\psi=-\Omega^2W\psi$, written in this basis, contains the matrix elements
\[
\langle W\phi_j,\phi_i\rangle.
\]
These matrix elements need not be diagonal when $H_K$ and $W$ do not commute. Consequently, the eigenfunctions of $H_K$ do not in general decouple the linearized mass-disordered equation and do not provide the normal-mode coordinates for the nonlinear problem.

Setting $H_0=-H_K$ and $E=\Omega^2$, we obtain
\begin{equation}\label{physicaloperator}
H_0\psi=EW\psi,\qquad
(H_0\psi)_n=-\psi_{n+1}-\psi_{n-1}+[2-K\cos(2\pi\bar u_n)]\psi_n.
\end{equation}
For the fixed analytic equilibrium, write its force profile as
\begin{equation}\label{physicalprofile}
V_{K,\rho}(t)=2-K\cos(2\pi h_{K,\rho}(t)).
\end{equation}
Set
\[
v_{K,\rho}(t)=K^{-1}V_{K,\rho}(t)=\frac{2}{K}-\cos(2\pi h_{K,\rho}(t)).
\]
The normal-mode equation \eqref{physicaloperator} is a case of \eqref{weighted}, with $\lambda=K$, $v=v_{K,\rho}$, $(x,y)=(\theta,y_0)$ and $E=\Omega^2$.

\subsection{High-energy localization at a fixed equilibrium}
Aubry's transition relates the loss of regularity of the equilibrium hull to the pinning of the chain \cite{A83,PA83}. To compare the role of the parameters in the unit-mass and mass-modulated chains, first set $m_n\equiv1$ in \eqref{physicaloperator}. This gives
\begin{equation}\label{unweightedphonon}
\psi_{n+1}+\psi_{n-1}+K\cos(2\pi\bar u_n)\psi_n=E_{\rm ph}\psi_n,
\qquad E_{\rm ph}=2-\Omega^2.
\end{equation}
In this equation, the same parameter $K$ determines both the equilibrium and the amplitude of the sampled potential. Increasing $K$ therefore changes the strength and the regularity of the phonon potential together. To separate these effects, Ketoja and Satija \cite{KS98,KS99} studied the two-parameter family
\begin{equation}\label{twoparam}
\psi_{n+1}+\psi_{n-1}+\lambda\cos(2\pi\bar u_n^{(K)})\psi_n
=E_{\rm ph}\psi_n,
\end{equation}
in which $K$ fixes the equilibrium and $\lambda$ independently controls the amplitude. The physical unit-mass FK equation is recovered on the diagonal $\lambda=K$.

This separation allows the equilibrium to remain analytic while the sampled potential becomes strong. The numerical phase diagram of Ketoja and Satija \cite[Fig.~2]{KS98}, computed at the upper spectral edge, contains exponentially localized states with $K<K_c$. Tong, Li and Hu \cite{TLH02} studied the corresponding electronic model and numerically found localization as $\lambda$ increased at fixed subcritical $K$. For a fixed nonconstant real analytic sampling function, rigorous large-coupling Anderson localization holds for each fixed phase and almost every driving frequency \cite[Theorem~1 and Comments~(1)--(2)]{BG00}. The large-amplitude regime thus provides a route to localization with an analytic potential, but increasing $\lambda$ independently of $K$ does not establish the corresponding spectral conclusion on the physical diagonal $\lambda=K$. Moreover, in the general analytic theorem the sampling function is fixed while the driving frequency varies, whereas the FK hull and its sampling frequency are determined together by the equilibrium.

For the chain with quasi-periodic masses considered here, Theorem~\ref{mainhigh} gives high-energy localization driven by mass disorder, with the substrate strength $K$, mean spacing $\rho$ and analytic equilibrium fixed. The conclusion holds for each fixed admissible phase and almost every mass frequency in the corresponding Diophantine slice. More precisely, we have the following corollary.

\begin{corollary}[High-energy localization for FK phonons]\label{fkcor}
Let $w$ satisfy the weight assumptions of Theorem~\ref{mainhigh}. Fix $K>0$, $c_0>0$, $A\geq3$ and a mean spacing $\rho$ satisfying \eqref{rhodc}. Consider an analytic FK equilibrium
\[
\bar u_n=h_{K,\rho}(\theta+n\rho),\qquad
h_{K,\rho}(t)=t+g_{K,\rho}(t),
\]
where $g_{K,\rho}$ is real analytic and $1$-periodic, and $h_{K,\rho}$ is strictly increasing and satisfies
\begin{equation}\label{fkhullequation}
h_{K,\rho}(t+\rho)+h_{K,\rho}(t-\rho)-2h_{K,\rho}(t)
+\frac K{2\pi}\sin(2\pi h_{K,\rho}(t))=0.
\end{equation}
Set $v_{K,\rho}(t)=2/K-\cos(2\pi h_{K,\rho}(t))$. For each fixed $(\theta,y_0)\in\T^2$ with $w(y_0)>0$, and for almost every $\alpha\in\DC^\rho_{c_0,A}(y_0)$, the self-adjoint weighted phonon operator in $\ell^2_m(\Z)$,
\begin{equation}\label{fkcoroperator}
\begin{gathered}
[H_{\rm FK}\psi]_n=m_n^{-1}\bigl(-\psi_{n+1}-\psi_{n-1}
+[2-K\cos(2\pi\bar u_n)]\psi_n\bigr),\\
m_n=w(y_0+n\alpha),\qquad
\|\psi\|_m^2=\sum_{n\in\Z}m_n|\psi_n|^2,
\end{gathered}
\end{equation}
has pure point spectral type on $[E_{\rm FK},\infty)$, where
\[
E_{\rm FK}=E_*(v_{K,\rho},w,c_0,A),
\]
and the corresponding eigenvectors decay exponentially in the displacement coordinates $\psi_n$. The operator is nonnegative and its spectrum is unbounded above. The physical spectral parameter is $E=\Omega^2$, where $\Omega\geq0$ denotes the vibration angular frequency. The exceptional set may depend on the fixed equilibrium, $\theta,y_0,w,c_0,A$, but is independent of $E$ and the eigenvector.
\end{corollary}
\begin{proof}
Fix the analytic equilibrium $(K,\rho,h_{K,\rho})$ in the statement. The degree-one identity $h_{K,\rho}(t+1)=h_{K,\rho}(t)+1$ makes
\[
v_{K,\rho}(t)=\frac{2}{K}-\cos(2\pi h_{K,\rho}(t))
\]
periodic. Analyticity is preserved by composition, and compactness of $\T$ gives a holomorphic extension to a strip of positive width. Thus this fixed sampling function satisfies the assumptions of Theorem~\ref{mainhigh}.

Equation~\eqref{physicaloperator} identifies the weighted phonon operator with \eqref{weighted} for $v=v_{K,\rho}$, $\lambda=K$ and $(x,y)=(\theta,y_0)$. With $E=\Omega^2$, Theorem~\ref{mainhigh} gives pure point spectral type and exponential decay above $E_{\rm FK}=E_*(v_{K,\rho},w,c_0,A)$ for almost every $\alpha$ in the stated slice.

It remains to prove nonnegativity. Write $h=h_{K,\rho}$ and $a=h'$. Since $h$ is increasing, $a\geq0$, and the degree-one condition gives $\int_\T a=1$. Differentiating \eqref{fkhullequation}, we obtain
\begin{equation}\label{fkpositivehull}
[2-K\cos(2\pi h(t))]a(t)=a(t+\rho)+a(t-\rho).
\end{equation}
If $a(t_0)=0$, nonnegativity implies $a(t_0+\rho)=a(t_0-\rho)=0$. Iteration and the density of the irrational rotation orbit would give $a\equiv0$, a contradiction. Therefore $a>0$.

Put $a_n=a(\theta+n\rho)$. For every finitely supported sequence $\psi$, the force matrix $H_{0,\rm FK}$ satisfies
\[
\langle H_{0,\rm FK}\psi,\psi\rangle
=\sum_n a_na_{n+1}
\left|\frac{\psi_{n+1}}{a_{n+1}}-\frac{\psi_n}{a_n}\right|^2\geq0.
\]
Its weighted realization and the unitarily equivalent Jacobi operator are therefore nonnegative on their cores, and hence on their self-adjoint closures from Section~2. The unboundedness noted in the introduction, together with nonnegativity, yields $\sup\Spec H_{\rm FK}=+\infty$.
\end{proof}

The phonon operator is nonnegative and has spectrum unbounded above, so the result applies to an unbounded part of its physical spectrum. Since $E=\Omega^2$, it gives spatially localized normal modes at high vibration frequencies while the equilibrium remains fixed.

\paragraph{Relation to nonlinear vibrations.}
We now return to the nonlinear equation \eqref{wmotion}.
When one linearizes the full nonlinear equation around a small-amplitude state $q^*$, the linearized force also contains the derivative $DR(q^*)$. Since $R$ vanishes to second order at the origin, this extra term is of order $q^*$ and is small for small-amplitude vibrations. Therefore the spectral properties of the principal phonon operator provide the basic normal-mode structure for finding breathers.

Localized linear normal modes are particularly important in the study of spatially localized vibrations of the nonlinear equation \eqref{wmotion}. On the high-energy spectral subspace supplied by Corollary~\ref{fkcor}, one can expand a displacement as
\[
q(t)=\sum_{j:E_j\geq E_{\rm FK}}a_j(t)\psi_j,
\]
where the $\psi_j$ form an orthonormal basis of this spectral subspace in the mass-weighted space. Their exponential decay provides localized coordinates for perturbative studies of nonlinear vibrations built from these localized modes. Extending this linear description to nonlinear localized vibrations requires further estimates on the mode interactions and the unbounded factors $m_n^{-1}$.

\paragraph{Acknowledgment.} The work was supported by the National Natural Science Foundation of China (Grant No. 12371189).
\section*{Ethics Declarations}
\noindent\textbf{Data Availability:} No datasets were generated or analyzed during the current study.

\noindent\textbf{Conflict of Interest:} The authors declare no conflict of interest to disclose.
\begingroup\small\linespread{1}\selectfont
\interlinepenalty=10000
\printbibliography

@article{AD83,
  shorthand = {AD83},
  sortkey = {001},
  author = {Aubry, S. and Le Daeron, P. Y.},
  title = {The discrete {Frenkel--Kontorova} model and its extensions: {I}. {Exact} results for the ground-states},
  journal = {Physica D},
  volume = {8},
  number = {3},
  pages = {381--422},
  year = {1983}
}

@article{A83,
  shorthand = {A83},
  sortkey = {002},
  author = {Aubry, S.},
  title = {The twist map, the extended {Frenkel--Kontorova} model and the devil's staircase},
  journal = {Physica D},
  volume = {7},
  number = {1--3},
  pages = {240--258},
  year = {1983}
}

@article{BG00,
  shorthand = {BG00},
  sortkey = {003},
  author = {Bourgain, J. and Goldstein, M.},
  title = {On nonperturbative localization with quasi-periodic potential},
  journal = {Ann. of Math.},
  series = {2},
  volume = {152},
  number = {3},
  pages = {835--879},
  year = {2000}
}

@article{BGS02,
  shorthand = {BGS02},
  sortkey = {004},
  author = {Bourgain, J. and Goldstein, M. and Schlag, W.},
  title = {{Anderson} localization for {Schr\"odinger} operators on {$\mathbb{Z}^2$} with quasi-periodic potential},
  journal = {Acta Math.},
  volume = {188},
  number = {1},
  pages = {41--86},
  year = {2002}
}

@article{BK19,
  shorthand = {BK19},
  sortkey = {005},
  author = {Bourgain, J. and Kachkovskiy, I.},
  title = {{Anderson} localization for two interacting {\mbox{quasiperiodic}} particles},
  journal = {Geom. Funct. Anal.},
  volume = {29},
  number = {1},
  pages = {3--43},
  year = {2019}
}

@article{BK91,
  shorthand = {BK91},
  sortkey = {006},
  author = {Braun, O. M. and Kivshar, Y. S.},
  title = {Nonlinear dynamics of the {Frenkel--Kontorova} model with impurities},
  journal = {Phys. Rev. B},
  volume = {43},
  pages = {1060--1073},
  year = {1991}
}

@book{BK04,
  shorthand = {BK04},
  sortkey = {007},
  author = {Braun, O. M. and Kivshar, Y. S.},
  title = {The {Frenkel--Kontorova} Model: Concepts, Methods, and Applications},
  series = {Texts and Monographs in Physics},
  publisher = {Springer-Verlag},
  address = {Berlin},
  year = {2004}
}

@article{Bou02,
  shorthand = {Bou02},
  sortkey = {008},
  author = {Bourgain, J.},
  title = {Estimates on {Green}'s functions, localization and the quantum kicked rotor model},
  journal = {Ann. of Math.},
  series = {2},
  volume = {156},
  number = {1},
  pages = {249--294},
  year = {2002}
}

@book{Bou05,
  shorthand = {Bou05},
  sortkey = {009},
  author = {Bourgain, J.},
  title = {Green's Function Estimates for Lattice {Schr\"odinger} Operators and Applications},
  series = {Annals of Mathematics Studies},
  volume = {158},
  publisher = {Princeton University Press},
  address = {Princeton, NJ},
  year = {2005}
}

@article{CSZ24,
  shorthand = {CSZ24},
  sortkey = {010},
  author = {Cao, H. and Shi, Y. and Zhang, Z.},
  title = {Quantitative {Green}'s function estimates for lattice quasi-periodic {Schr\"odinger} operators},
  journal = {Sci. China Math.},
  volume = {67},
  number = {5},
  pages = {1011--1058},
  year = {2024}
}

@book{DF22,
  shorthand = {DF22},
  sortkey = {011},
  author = {Damanik, D. and Fillman, J.},
  title = {One-Dimensional Ergodic {Schr\"odinger} Operators {I}: General Theory},
  series = {Graduate Studies in Mathematics},
  volume = {221},
  publisher = {American Mathematical Society},
  address = {Providence, RI},
  year = {2022}
}

@article{Ge24,
  shorthand = {Ge24},
  sortkey = {012},
  author = {Ge, L.},
  title = {On the almost reducibility conjecture},
  journal = {Geom. Funct. Anal.},
  volume = {34},
  number = {1},
  pages = {32--59},
  year = {2024}
}

@article{GYZ23,
  shorthand = {GYZ23},
  sortkey = {013},
  author = {Ge, L. and You, J. and Zhou, Q.},
  title = {Exponential dynamical localization: criterion and applications},
  journal = {Ann. Sci. \'{E}c. Norm. Sup\'{e}r.},
  series = {4},
  volume = {56},
  number = {1},
  pages = {91--126},
  year = {2023}
}

@article{Han19,
  shorthand = {Han19},
  sortkey = {014},
  author = {Han, R.},
  title = {{Shnol}'s theorem and the spectrum of long range operators},
  journal = {Proc. Amer. Math. Soc.},
  volume = {147},
  number = {7},
  pages = {2887--2897},
  year = {2019}
}

@article{JL17,
  shorthand = {JL17},
  sortkey = {015},
  author = {Jitomirskaya, S. and Liu, W.},
  title = {Arithmetic spectral transitions for the {Maryland} model},
  journal = {Comm. Pure Appl. Math.},
  volume = {70},
  number = {6},
  pages = {1025--1051},
  year = {2017}
}

@article{JLS20,
  shorthand = {JLS20},
  sortkey = {016},
  author = {Jitomirskaya, S. and Liu, W. and Shi, Y.},
  title = {{Anderson} localization for multi-frequency quasi-periodic operators on {$\mathbb{Z}^d$}},
  journal = {Geom. Funct. Anal.},
  volume = {30},
  number = {2},
  pages = {457--481},
  year = {2020}
}

@article{JSY21,
  shorthand = {JSY21},
  sortkey = {017},
  author = {Jian, W. and Shi, J. and Yuan, X.},
  title = {{Anderson} localization for long-range operators with singular potentials},
  journal = {J. Math. Phys.},
  volume = {62},
  number = {2},
  pages = {022703},
  year = {2021}
}

@article{JY17,
  shorthand = {JY17},
  sortkey = {018},
  author = {Jitomirskaya, S. and Yang, F.},
  title = {Singular continuous spectrum for singular potentials},
  journal = {Comm. Math. Phys.},
  volume = {351},
  number = {3},
  pages = {1127--1135},
  year = {2017}
}

@misc{KS98,
  shorthand = {KS98},
  sortkey = {019},
  author = {Ketoja, J. A. and Satija, I. I.},
  title = {Hidden dimer in the {Frenkel--Kontorova} model},
  howpublished = {Preprint, arXiv:cond-mat/9802149},
  year = {1998}
}

@article{KS99,
  shorthand = {KS99},
  sortkey = {020},
  author = {Ketoja, J. A. and Satija, I. I.},
  title = {Correlated disorder and propagating modes in the {Frenkel--Kontorova} model},
  journal = {Phys. Rev. B},
  volume = {59},
  number = {14},
  pages = {9174--9184},
  year = {1999}
}

@article{PA83,
  shorthand = {PA83},
  sortkey = {021},
  author = {Peyrard, M. and Aubry, S.},
  title = {Critical behaviour at the transition by breaking of analyticity in the discrete {Frenkel--Kontorova} model},
  journal = {J. Phys. C: Solid State Phys.},
  volume = {16},
  number = {9},
  pages = {1593--1608},
  year = {1983}
}

@article{Spe10,
  shorthand = {Spe10},
  sortkey = {022},
  author = {Spencer, T.},
  title = {Mathematical aspects of {Anderson} localization},
  journal = {Internat. J. Modern Phys. B},
  volume = {24},
  number = {12--13},
  pages = {1621--1639},
  year = {2010}
}

@book{Tes00,
  shorthand = {Tes00},
  sortkey = {023},
  author = {Teschl, G.},
  title = {Jacobi Operators and Completely Integrable Nonlinear Lattices},
  series = {Mathematical Surveys and Monographs},
  volume = {72},
  publisher = {American Mathematical Society},
  address = {Providence, RI},
  year = {2000}
}

@article{TLH02,
  shorthand = {TLH02},
  sortkey = {024},
  author = {Tong, P. and Li, B. and Hu, B.},
  title = {Electronic properties of the one-dimensional {Frenkel--Kontorova} model},
  journal = {Phys. Rev. Lett.},
  volume = {88},
  number = {4},
  pages = {046804},
  year = {2002}
}

@article{Zha24,
  shorthand = {Zha24},
  sortkey = {025},
  author = {Zhang, X.},
  title = {{Anderson} localization for block {Jacobi} operators with quasi-periodic meromorphic potential},
  journal = {Math. Methods Appl. Sci.},
  volume = {47},
  number = {16},
  pages = {12816--12832},
  year = {2024}
}
\endgroup

\end{document}